\documentclass[a4paper]{article}
\usepackage{amsmath,amsthm,amsfonts,thmtools}
\usepackage{amssymb}
\usepackage[cmintegrals]{newtxmath}
\usepackage{graphicx}
\usepackage{listings}
\usepackage{subfigure}
\usepackage{euscript}
\usepackage{units}
\usepackage{enumerate}
\usepackage{enumitem}
\usepackage{xcolor}
\usepackage{todonotes}
\usepackage{booktabs}
\usepackage{tikz}

\usepackage[bookmarks=true,hidelinks]{hyperref}
\usepackage{bbm}
\usepackage{multirow}

\numberwithin{equation}{section}
\newtheorem{theorem}{Theorem}[section]

\newtheorem{lemma}[theorem]{Lemma}

\newtheorem{algorithm}[theorem]{Algorithm}

\numberwithin{equation}{section}

\theoremstyle{definition}

\newtheoremstyle{myremarkstyle}{}{}{}{}{\bfseries}{.}{ }{}
\theoremstyle{myremarkstyle}
\declaretheorem[name=Remark,qed=$\blacksquare$,numberlike=theorem]{remark}

\declaretheorem[name=Definition,qed=$\blacksquare$,numberlike=theorem]{definition}

\makeatletter
\newcommand*{\intavg}{%
  \mint@l{-}{}%
}
\newcommand*{\mint@l}[2]{%
  \@ifnextchar\limits{%
    \mint@l{#1}%
  }{%
    \@ifnextchar\nolimits{%
      \mint@l{#1}%
    }{%
      \@ifnextchar\displaylimits{%
        \mint@l{#1}%
      }{%
        \mint@s{#2}{#1}%
      }%
    }%
  }%
}
\newcommand*{\mint@s}[2]{%
  \@ifnextchar_{%
    \mint@sub{#1}{#2}%
  }{%
    \@ifnextchar^{%
      \mint@sup{#1}{#2}%
    }{%
      \mint@{#1}{#2}{}{}%
    }%
  }%
}
\def\mint@sub#1#2_#3{%
  \@ifnextchar^{%
    \mint@sub@sup{#1}{#2}{#3}%
  }{%
    \mint@{#1}{#2}{#3}{}%
  }%
}
\def\mint@sup#1#2^#3{%
  \@ifnextchar_{%
    \mint@sub@sup{#1}{#2}{#3}%
  }{%
    \mint@{#1}{#2}{}{#3}%
  }%
}
\def\mint@sub@sup#1#2#3^#4{%
  \mint@{#1}{#2}{#3}{#4}%
}
\def\mint@sup@sub#1#2#3_#4{%
  \mint@{#1}{#2}{#4}{#3}%
}
\newcommand*{\mint@}[4]{%
  \mathop{}%
  \mkern-\thinmuskip
  \mathchoice{%
    \mint@@{#1}{#2}{#3}{#4}%
        \displaystyle\textstyle\scriptstyle
  }{%
    \mint@@{#1}{#2}{#3}{#4}%
        \textstyle\scriptstyle\scriptstyle
  }{%
    \mint@@{#1}{#2}{#3}{#4}%
        \scriptstyle\scriptscriptstyle\scriptscriptstyle
  }{%
    \mint@@{#1}{#2}{#3}{#4}%
        \scriptscriptstyle\scriptscriptstyle\scriptscriptstyle
  }%
  \mkern-\thinmuskip
  \int#1%
  \ifx\\#3\\\else_{#3}\fi
  \ifx\\#4\\\else^{#4}\fi  
}
\newcommand*{\mint@@}[7]{%
  \begingroup
    \sbox0{$#5\int\m@th$}%
    \sbox2{$#5\int_{}\m@th$}%
    \dimen2=\wd0 %
    \let\mint@limits=#1\relax
    \ifx\mint@limits\relax
      \sbox4{$#5\int_{\kern1sp}^{\kern1sp}\m@th$}%
      \ifdim\wd4>\wd2 %
        \let\mint@limits=\nolimits
      \else
        \let\mint@limits=\limits
      \fi
    \fi
    \ifx\mint@limits\displaylimits
      \ifx#5\displaystyle
        \let\mint@limits=\limits
      \fi
    \fi
    \ifx\mint@limits\limits
      \sbox0{$#7#3\m@th$}%
      \sbox2{$#7#4\m@th$}%
      \ifdim\wd0>\dimen2 %
        \dimen2=\wd0 %
      \fi
      \ifdim\wd2>\dimen2 %
        \dimen2=\wd2 %
      \fi
    \fi
    \rlap{%
      $#5%
        \vcenter{%
          \hbox to\dimen2{%
            \hss
            $#6{#2}\m@th$%
            \hss
          }%
        }%
      $%
    }%
  \endgroup
}

\def\XXint#1#2#3{{\setbox0=\hbox{$#1{#2#3}{\int}$ }
		\vcenter{\hbox{$#2#3$ }}\kern-.6\wd0}}

\renewcommand{\geq}{\geqslant}
\renewcommand{\leq}{\leqslant}

\newcommand{\eps} {\varepsilon}

\renewcommand{\epsilon}{\varepsilon}
\renewcommand{\phi}{\varphi}

\newcommand{\R}{\mathbb{R}}
\newcommand{\N}{\mathbb{N}}

\newcommand{\U}{{\bf U}}		

\newcommand{\vel}{\mathbf{u}}
\newcommand{\n}{\mathbf{n}}

\newcommand{\F}{{\bf F}}

\newcommand{\cO}{{\mathcal O}}

\newcommand{\om}{\Omega}

\newcommand{\map}{\EuScript{L}}
\newcommand{\train}{\EuScript{S}}
\newcommand{\test}{\EuScript{T}}
\newcommand{\val}{\EuScript{V}}
\newcommand{\reg}{\EuScript{R}}
\newcommand{\er}{\EuScript{E}}

\newcommand{\cost}{\EuScript{C}}
\newcommand{\Jq}{\EuScript{J}_q}
\newcommand{\Ds}{\EuScript{D}}
\newcommand{\JL}{\EuScript{J}_L}
\newcommand{\Jd}{\EuScript{J}_d}

\begin{document}

\date{\today}

\title{Deep learning observables in computational fluid dynamics}

\author{Kjetil O. Lye \thanks{Seminar for Applied Mathematics (SAM), D-Math \newline
  ETH Z\"urich, R\"amistrasse 101, 
  Z\"urich-8092, Switzerland}, Siddhartha Mishra \thanks{Seminar for Applied Mathematics (SAM), D-Math \newline
  ETH Z\"urich, R\"amistrasse 101, 
  Z\"urich-8092, Switzerland} and Deep Ray \thanks{Department of Computational \& Applied Mathematics, Rice University, Houston, Texas, USA.}}

\date{\today}

\maketitle
\begin{abstract}
Many large scale problems in computational fluid dynamics such as uncertainty quantification, Bayesian inversion, data assimilation and PDE constrained optimization are considered very challenging computationally as they require a large number of expensive (forward) numerical solutions of the corresponding PDEs. We propose a machine learning algorithm, based on deep artificial neural networks, that predicts the underlying \emph{input parameters to observable} map from a few training samples (computed realizations of this map). By a judicious combination of theoretical arguments and empirical observations, we find suitable network architectures and training hyperparameters that result in robust and efficient neural network approximations of the parameters to observable map. Numerical experiments are presented to demonstrate low prediction errors for the trained network networks, even when the network has been trained with a few samples, at a computational cost which is several orders of magnitude lower than the underlying PDE solver. 

Moreover, we combine the proposed deep learning algorithm with Monte Carlo (MC) and Quasi-Monte Carlo (QMC) methods to efficiently compute uncertainty propagation for nonlinear PDEs. Under the assumption that the underlying neural networks generalize well, we prove that the deep learning MC and QMC algorithms are guaranteed to be faster than the baseline (quasi-) Monte Carlo methods. Numerical experiments demonstrating one to two orders of magnitude speed up over baseline QMC and MC algorithms, for the intricate problem of computing probability distributions of the observable, are also presented.

\end{abstract}
\section{Introduction}
Many interesting fluid flows are modeled by so-called \emph{convection-diffusion} equations \cite{LanLip1} i.e, nonlinear partial differential equations (PDEs) of the generic form,
\begin{equation}
\label{eq:cde}
\U_t + {\rm div}_x (\F(\U)) = \nu ~{\rm div}_x({\bf D}(\U) \nabla_x \U), \quad (x,t) \in D \subset \R^{d_s} \times \R_+,
\end{equation}
with $\U\in \R^m$ denoting the vector of unknowns, $\F = (F_i)_{1\leq i \leq d_s}$ the \emph{flux vector}, ${\bf D} = (D_{ij})_{1\leq i,j \leq d_s}$ the 
\emph{diffusion matrix} and $\nu$ a small scale parameter representing kinematic viscosity. 

Prototypical examples for \eqref{eq:cde} include the compressible Euler equations of gas dynamics, shallow water equations of oceanography and the magnetohydrodynamics (MHD) equations of plasma
physics. These PDEs are \emph{hyperbolic systems of conservation laws} i.e, special cases of \eqref{eq:cde} with  $\nu = 0$ and with the flux Jacobian $\partial_{\U} (\F\cdot n)$ having real eigenvalues for all normal vectors $n$. Another important example for \eqref{eq:cde} is provided by the incompressible Navier-Stokes equations, where $0 < \nu << 1$ and the flux function is \emph{non-local} on account of the divergence-free constraint on the velocity field. 

It is well known that solutions to convection-diffusion equations \eqref{eq:cde} can be very complicated. These solutions might include singularities such as shock waves and contact discontinuities in the case of hyperbolic systems of conservation laws.
For small values of $\nu$, these solutions can be generically unstable, even turbulent, and contain structures with a large range of spatio-temporal scales \cite{LanLip1}.

Numerical schemes play a key role in the study of fluid flows and a large variety of robust and efficient numerical methods have been designed to approximate them. These include (conservative) finite difference \cite{LEV2}, finite volume \cite{GR1,HEST1}, discontinuous Galerkin (DG) finite element \cite{HEST1} and spectral (viscosity) methods \cite{Ors1}. These methods have been extremely successful in practice and are widely used in science and engineering today.

The exponential increase in computational power in the last decades provides us with the opportunity to solve very challenging large scale problems in computational fluid dynamics, such as uncertainty quantification (UQ)  \cite{UQbook,UQhb}, (Bayesian) inverse problems \cite{STU1}  and real-time optimal control, design and PDE constrained (shape) optimization \cite{OC,OC1}. In such problems, one is not always interested in computing the whole solution field $\U$ of \eqref{eq:cde}. Rather and in analogy with experimental measurements, one is interested in computing the so-called \emph{observables} (functionals or quantities of interest) for the solution $\U$ of \eqref{eq:cde}. These can be expressed in the generic form,
\begin{equation}
\label{eq:obs1}
L(\U) = \int\limits_{D \times \R_+} \psi(x,t) g(\U(x,t)) dx dt.
\end{equation}
Here $\psi:D \times \R_+ \rightarrow \R$ and $g:\R^m \rightarrow \R$ are suitable test functions. Prototypical examples of such observables (functionals) are provided by body forces, such as the \emph{lift} and \emph{drag} in aerodynamic simulations, and by the runup height (at certain probe points) in simulations of tsunamis.

Moreover in practice, one is interested not just in a single value, but rather in the \emph{statistics} of such observables. Typical statistical quantities of interest are the mean, variance, higher-moments and probability density functions (pdfs) of the observable \eqref{eq:obs1}. Such statistical quantities quantify uncertainty in the solution, propagated from possible uncertainty in the underlying inputs i.e, fluxes, diffusion coefficients, initial and boundary data. They might also stem from the presence of a statistical spread in design parameters such as those describing the geometry of the computational domain. It is customary to represent the resulting solution field as $\U = \U(y)$, with $y \in Y \subset \R^{d}$, a possibly very high dimensional parameter space. Thus, the goal of many CFD simulations is to compute (statistics of) the so-called \emph{parameters to observable} map $y \mapsto L(\U(y))$. 

Calculating a single realization of this parameters to observable map might require a very expensive forward solve of \eqref{eq:cde} with a CFD solver and a quadrature to calculate \eqref{eq:obs1}. However, in UQ, Bayesian inversion or optimal design and control, one needs to evaluate a large number of instances of this map, necessitating a very high computational cost, even on state of the art HPC systems. As a concrete example, we consider a rather simple yet prototypical situation. We are interested in computing statistics of body forces such as lift and drag for a model two-dimensional RAE2822 airfoil \cite{UMRIDA} (see figure \ref{fig:rae_mesh}). A single forward solve of the underlying compressible Euler equations on this airfoil geometry on a mesh of high resolution (see figure \ref{fig:rae_mesh}), with a state of the art high-resolution finite volume solver \cite{RCFM}, takes approximately $7$ wall-clock hours on a HPC cluster (see table \ref{tab:afoilcost}). This cost scales up in three space dimensions proportionately. However, a typical UQ problem such as Bayesian inversion with a state of the art MCMC algorithm, for instance the Random walk Metropolis-Hastings algorithm \cite{STU1}, might need upto $10^5-10^6$ such realizations, rendering even a two-dimensional Bayesian inversion infeasible ! This example illustrates the fact that the high-computational cost of the parameters to observable map makes problems such as forward UQ, Bayesian inversion, and optimal control and design very challenging. Hence, \emph{we need numerical methods which allow ultrafast (several order of magnitude faster than state of the art) computations of the underlying parameters to observable map}.

Machine learning, in the form of artificial neural networks (ANNs), has become extremely popular in computer science in recent years. This term is applied to methods that aim to approximate functions with layers of
units (neurons), connected by (affine) linear operations between units and nonlinear activations within units, \cite{DLbook} and references therein. \emph{Deep learning}, i.e an artificial neural network with a large number of intermediate (hidden) layers has proven extremely successful at diverse tasks, for instance in image processing, computer vision, text and speech recognition, game intelligence and more recently in protein folding \cite{AFOL}, see \cite{DL-nat} and references therein for more applications of deep learning. A key element in deep learning is the \emph{training} of tunable parameters in the underlying neural network by (approximately) minimizing suitable \emph{loss functions}. The resulting (non-convex) optimization problem, on a very high dimensional underlying space, is customarily solved with variants of the stochastic gradient descent method \cite{SG}.

Deep learning is being increasingly used in the context of numerical solution of partial differential equations. Given that neural networks are very powerful universal function approximators \cite{Cy1,Kor1,Bar1,MP1,YAR1}, it is natural to consider the space of neural networks as an ansatz space for approximating solutions of PDEs. First proposed in \cite{Lag1} on an underlying collocation approach, it has been successfully used recently in different contexts in \cite{Kar1,Kar2,JR1,E1,E2,E3} and references therein. Given spatial and temporal locations as inputs, the (deep) neural networks, proposed by these authors, approximate the solution of the underlying PDE, by outputting function values. This approach appears to work quite well for problems with high regularity (smoothness) of the underlying solutions (see \cite{SZ1,E3}) and/or if the solution of the underlying PDE possesses a representation formula in terms of integrals \cite{E2,E3}.  Given their low regularity, it is unclear if solutions of \eqref{eq:cde}, realized as functions of space and time, can be efficiently learned by deep neural networks \cite{INC}.

Several papers applying deep learning techniques in the context of CFD advocate embedding deep learning modules within existing CFD codes to increase their efficiency \cite{INC,DR1,DL_SM1}. In this paper,  we adopt a different approach and propose to use \emph{deep neural networks} to \emph{learn and predict the input parameters to observable map}. Our algorithm will be based on fully connected networks (multi-layer preceptrons) which output values of the observable \eqref{eq:obs1} for different input parameters $y \in Y \subset \R^{d}$ with $d >>1$. The network will be trained on data, generated from a few, say ${\mathcal  O}(100)$, \emph{samples} i.e, realizations of \eqref{eq:obs1} with expensive CFD solvers.

However, the task of designing deep neural networks that will approximate the parameters to observable map with reasonable accuracy is far from straightforward on account of the following issues,
\begin{itemize}
\item Approximation results for deep neural networks \cite{YAR1,Pet1,Pet2}, stipulate a network size of ${\mathcal O}\left(\epsilon^{\frac{-d}{s}}\right)$ for attaining an error of ${\mathcal O}(\epsilon)$, with $s$ denoting the Sobolev regularity of the underlying map. However, and as explained in section \ref{sec:theo}, the parameters to observable map for most CFD problems is atmost Lipschitz continuous. Thus, we might require unrealistically large networks for approximating the underlying maps in this context. 
\item Another relevant metric for estimating errors with deep neural network is the so-called \emph{generalization error} of the trained network, (see section \ref{sec:theo}), that measures the errors of the network in predicting \emph{unseen data}. The best available bounds for this error scale (inversely) with the square-root of the number of training samples. However, in this context, our training samples are generated from expensive CFD simulations. Hence, we might end up with large prediction errors in this \emph{data poor} regime. 
\end{itemize}
Hence, \emph{it is quite challenging to find deep neural networks that can accurately approximate maps of low Sobolev regularity in a data poor regime}. Moreover, there are several hyperparameters that need to be specified in this framework, for instance size of the networks, choice of which variant of the stochastic gradient algorithm that one uses, choice of loss functions and regularizations thereof etc. A priori, the prediction errors can be sensitive to these choices.  We propose an \emph{ensemble training} procedure to search the hyperparameter space systematically and identify network architectures that are efficient in our context. By a combination of theoretical considerations and rigorous empirical experimentation, we provide a recipe for finding appropriate deep neural networks to compute the parameters to observable map i.e, those network architectures which ensure a low generalization error, even for relatively few training samples. 

A second aim of this paper is to employ these trained neural networks, in conjunction with (Quasi-)Monte Carlo methods to compute statistical quantities of interest, particularly the underlying probability distributions of the observable and to demonstrate that these \emph{deep learning (quasi)-Monte Carlo} methods significantly outperform the baseline algorithms. 

The rest of the paper is organized as follows: the deep learning algorithm is presented in section \ref{sec:2}. The deep learning (Quasi-)Monte Carlo algorithm for forward UQ is presented in section \ref{sec:3} and some details for the implementation of both sets of algorithms are provided in section \ref{sec:imp}. In section \ref{sec:num}, we present numerical experiments illustrating the proposed algorithms and the results of the paper are summarized and discussed in section \ref{sec:disc}.
\section{Deep learning algorithm}
\label{sec:2}
\subsection{The problem}
\label{sec:2.1}
We consider the following very general form of a \emph{parameterized} convection-diffusion equation,
\begin{equation}
\label{eq:cdep}
\begin{aligned}
\partial_t \U(t,x,y) + {\rm div}_x (\F(y, \U)) &= \nu ~{\rm div}_x({\bf D}(y,\U) \nabla_x \U), \quad \forall~(t,x,y) \in [0,T] \times D(y) \times Y, \\
\U(0,x,y) &= \overline{\U}(x,y), \quad \forall~ (x,y) \in D(y) \times Y, \\
L_{b} \U(t,x,y) &= \U_b (t,x,y), \quad \forall ~(t,x,y) \in [0,T] \times D(y) \times Y
\end{aligned}
\end{equation}
Here, $Y$ is the parameter space and without loss of generality, we assume it to be $Y = [0,1]^{d}$, for some $d \in {\mathbb N}$.  

The spatial domain is labeled as $y \rightarrow D(y) \subset \R^{d_s}$ and  $\U: [0,T] \times D \times Y  \rightarrow \R^m$ is the vector of unknowns. The flux vector is denoted as $\F = (F_i)_{1\leq i \leq d_s}:\R^m \times Y \rightarrow \R^m$ and ${\bf D} = (D_{ij})_{1\leq i,j \leq d_s}: \R^m \rightarrow \R^m$ is the \emph{diffusion matrix}. 

The operator $L_b$ is a \emph{boundary} operator that imposes boundary conditions for the PDE, for instance the no-slip boundary condition for incompressible Navier-Stokes equations or characteristic boundary conditions for hyperbolic systems of conservation laws. Additional conditions such as hyperbolicity for the flux function $\F$ and positive-definiteness of the diffusion matrix ${\bf D}$ might also be imposed, depending on the specific problem that is being considered. 

We remark that the parameterized PDE \eqref{eq:cdep} will arise in the context of (both forward and inverse) UQ, when the underlying convection-diffusion equation \eqref{eq:cde} contains uncertainties in the domain, the flux and diffusion coefficients and in the initial and boundary data. This input uncertainty propagates into the solution. Following \cite{UQbook} and references therein, it is customary to model such random inputs and the resulting solution uncertainties by \emph{random fields}. Consequently, one can parameterize the probability space on which the random field is defined in terms of a parameter space $Y \subset \R^{d}$, for instance by expressing random fields in terms of (truncated) Karhunen-Loeve expansions. By normalizing the resulting random variables, one may assume $Y = [0,1]^{d}$, with possibly a large value of the parameter dimension $d$. Moreover, there exists a measure, $\mu \in {\rm Prob}(Y)$, with respect to which the data from the underlying parameter space is drawn.

We point out that the above framework is also relevant in problems of optimal control, design and PDE constrained optimization. In these problems, the parameter space $Y$ represents the set of design or control parameters. 

For the parameterized PDE \eqref{eq:cdep}, we aim to compute observables of the following general form,
\begin{equation}
\label{eq:obsp}
L_g(y,\U) := \int\limits_0^T\int\limits_{D_y} \psi(x,t) g(\U(t,x,y)) dx dt, \quad {\rm for} ~\mu~{\rm a.e}~y \in Y.
\end{equation} 
Here, $\psi \in L^1_{{\rm loc}} (D_y \times (0,T))$ is a  \emph{test function} and $g \in C^s(\R^m)$, for $s \geq 1$. Most interesting observables encountered in experiments, such as the lift and the drag, can be cast in this general form. 

For fixed functions $\psi,g$, we also define the \emph{parameters to observable} map:
\begin{equation}
\label{eq:ptoob}
\map:y \in Y \rightarrow \map(y) = L_g(y,\U),
\end{equation}
with $L_g$ being defined by \eqref{eq:obsp}. 

We also assume that there exist suitable numerical schemes for approximating the convection-diffusion equation \eqref{eq:cdep} for every parameter vector $y \in Y$. These schemes could be of the finite difference, finite volume, DG or spectral
type, depending on the problem and on the baseline CFD code. Hence for any mesh parameter (grid size, time step) $\Delta$, we are assuming that for any parameter vector $y \in Y$, a high-resolution approximate solution $\U^{\Delta}(y) \approx \U(y)$ is available. Hence, there exists an approximation to the \emph{input to observable} map $\map$ of the form,
 \begin{equation}
\label{eq:ptoob1}
\map:y \in Y \rightarrow \map(y) = L_g(y,\U^\Delta),
\end{equation}
with the integrals in \eqref{eq:obsp} being approximated to high accuracy by quadratures. Therefore, the original input parameters to observable map $\map$ is approximated by $\map^{\Delta}$ to very high accuracy i.e, for every value of a tolerance $\epsilon > 0$, there exists a $\Delta << 1$, such that
\begin{equation}
\label{eq:err1}
\|\map(y) - \map^{\Delta}(y)\|_{L_{\mu}^p(Y)} < \epsilon,
\end{equation}
for some $1 \leq p \leq \infty$ and weighted norm, 
$$
\|f\|_{L_{\mu}^p(Y)} := \left(\int_Y |f(y)|^p d\mu(y) \right)^{\frac{1}{p}},
$$
for $1 \leq p < \infty$. The $L^{\infty}_{\mu}$ norm is analogously defined.  

\subsection{Deep learning the parameters to observable map}
\label{sec:2.2}
As stated earlier, it can be very expensive to compute the map $\map^{\Delta}(y)$ for each single realization, $y \in Y$, as a high-resolution CFD solver, possibly entailing a very large number of 
degrees of freedom, needs to be used. We propose instead, to  \emph{learn} this map by deep neural networks. This process entails the following steps,
 \subsubsection{Training set.} 
 \label{sec:training_set}
 As is customary in supervised learning \cite{DLbook} and references therein, we need to generate or obtain data to train the network. To this end, we select a set of parameters $\EuScript{S} = \{y_i\}_{1 \leq i \leq N}$, with each $y_i \in Y$. The points in $\EuScript{S}$ can be chosen as,
\begin{itemize}
\item [(i.)] randomly from $Y$, independently and identically distributed with the underlying probability distribution $\mu$. 
\item [(ii.)] from a suitable set of quadrature points in $Y$, for instance the so-called \emph{low discrepancy sequences} that arise in Quasi-Monte Carlo (QMC) quadrature algorithms \cite{CAF1}. Examples of such sequences include Sobol or Halton QMC quadrature points \cite{CAF1}.
\end{itemize}
We emphasize that the generation of QMC points is very cheap, particularly for Sobol or Halton sequences. Moreover, these points are better spread out over the parameter space than a random selection of points and might provide more detailed information about it \cite{CAF1}. Hence, a priori QMC points might be a better choice for sampling the data. One can also replace QMC points with other hierarchical algorithms such as nodes of sparse grids (Smolyak quadrature points) \cite{GRE1} to form the set $\train$.

Once the training parameter set $\train$ is chosen, we perform a set of high-resolution CFD simulations to obtain $\map^{\Delta} (y)$, for all $y \in \train$. As each high-resolution CFD simulation could be very expensive, we will require that $N = \# (\train)$ will not be very large. It will typically be of at most ${\mathcal O}(100)$. 
\begin{figure}[htbp]
\centering
\includegraphics[width=8cm]{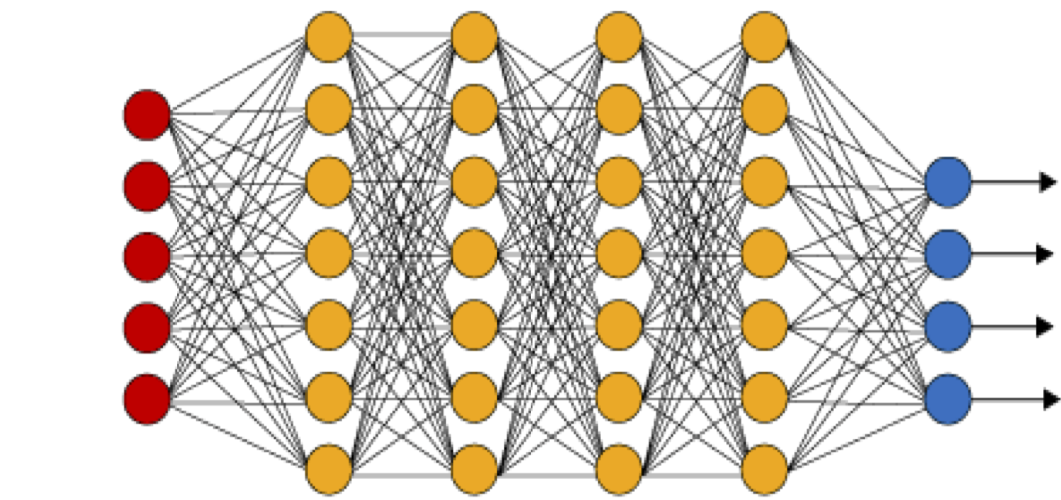}
\caption{An illustration of a (fully connected) deep neural network. The red neurons represent the inputs to the network and the blue neurons denote the output layer. They are
connected by hidden layers with yellow neurons. Each hidden unit (neuron) is connected by affine linear maps between units in different layers and then with nonlinear (scalar) activation functions within units.}
\label{fig:1}
\end{figure}

\subsubsection{Neural network.} 
\label{sec:NN}
Given an input vector, a feedforward neural network (also termed as a multi-layer perceptron), shown in figure \ref{fig:1}, consists of layer of units (neurons) which compose of either affine-linear maps between units (in successive layers) or scalar non-linear activation functions within units, culminating in an output \cite{DLbook}. In our framework, for any input vector $z \in Y$, we represent an artificial neural network as,
\begin{equation}
\label{eq:ann1}
\map_{\theta}(z) = C_K \odot\sigma \odot C_{K-1}\ldots \ldots \ldots \odot\sigma \odot C_2 \odot \sigma \odot C_1(z).
\end{equation} 
Here, $\odot$ refers to the composition of functions and $\sigma$ is a scalar (non-linear) activation function. A large variety of activation functions have been considered in the machine learning literature \cite{DLbook}. A very popular choice is the \emph{ReLU} function,
\begin{equation}
\label{eq:relu}
\sigma(z) = \max(z,0).
\end{equation}
When, $z \in \R^p$ for some $p > 1$, then the output of the ReLU function in \eqref{eq:relu} is evaluated componentwise. 

Moreover, for any $1 \leq k \leq K$, we define
\begin{equation}
\label{eq:C}
C_k z_k = W_k z_k + b_k, \quad {\rm for} ~ W_k \in \R^{d_{k+1} \times d_k}, z_k \in \R_{d_k}, b_k \in \R^{d_{k+1}}.
\end{equation}
For consistency of notation, we set $d_1 = d$ and $d_K = 1$. 

Thus in the terminology of machine learning (see also figure \ref{fig:1}), our neural network \eqref{eq:ann1} consists of an input layer, an output layer and $(K-1)$ hidden layers for some $1 < K \in \N$. The $k$-th hidden layer (with $d_k$ neurons) is given an input vector $z_k \in \R^{d_k}$ and transforms it first by an affine linear map $C_k$ \eqref{eq:C} and then by a ReLU (or another) nonlinear (component wise) activation $\sigma$ \eqref{eq:relu}. Although the neural network consists of composition of very elementary functions, its complexity and ability to learn very general functions arises from the interactions between large number of hidden layers \cite{DLbook}.  

A straightforward addition shows that our network contains $\left(d + 1 + \sum\limits_{k=2}^{K-1} d_k\right)$ neurons. 

We denote, 
\begin{equation}
\label{eq:theta}
\theta = \{W_k, b_k\}, \theta_W = \{ W_k \}\quad \forall~ 1 \leq k \leq K,
\end{equation} 
to be the concatenated set of (tunable) weights for our network. It is straightforward to check that $\theta \in \Theta \subset \R^M$ with
\begin{equation}
\label{eq:ns}
M = \sum\limits_{k=1}^{K-1} (d_k +1) d_{k+1}.
\end{equation}
 Thus, depending on the dimensions of the input parameter vector and the number (depth) and size (width) of the hidden layers, our proposed neural network can contain a large number of weights. Moreover, the neural network explicitly depends on the choice of the weight vector $\theta \in \Theta$, justifying the notation in \eqref{eq:ann1}. 
 
 Although a variety of network architectures, such as convolutional neural networks or recurrent neural networks, have been proposed in the machine learning literature, \cite{DLbook} and references therein, we will restrict ourselves to fully connected architectures i.e, we do not a priori assume any \emph{sparsity} structure for our set $\Theta$. 

\subsubsection{Loss functions and optimization.} 
For any $y \in \train$, we have already evaluated $\map^{\Delta}(y)$ from the high-resolution CFD simulation. One can readily compute the output of the neural network $\map_{\theta} (y)$ for any weight vector $\theta \in \Theta$. We define the so-called 
\emph{loss function} or mismatch function, as 
\begin{equation}
\label{eq:lf1}
J (\theta) : = \sum\limits_{y \in \train} |\map^{\Delta}(y) - \map_{\theta} (y) |^p,
\end{equation}
for some $1 \leq p < \infty$.  

The goal of the training process in machine learning is to find the weight vector $\theta \in \Theta$, for which the loss function \eqref{eq:lf1} is minimized. The resulting optimization (minimization) problem might lead to searching a minimum of a non-convex loss function. So, it is not uncommon in machine learning \cite{DLbook} to regularize the minimization problem i.e we seek to find,
\begin{equation}
\label{eq:lf2}
\theta^{\ast} = {\rm arg}\min\limits_{\theta \in \Theta} \left(J(\theta) + \lambda \reg(\theta) \right).
\end{equation}  
Here, $\reg:\Theta \mapsto \R$ is a \emph{regularization} (penalization) term. A popular choice is to set  $\R(\theta) = \|\theta_W\|^q_q$ for either $q=1$ (to induce sparsity) or $q=2$. The parameter $0 < \lambda << 1$ balances the regularization term with actual loss $J$ \eqref{eq:lf1}. 

The above minimization problem amounts to finding a minimum of a possibly non-convex function over a subset of $\R^M$ for very large $M$. We can approximate the solutions to this minimization problem iteratively, either by a full batch gradient descent algorithm or by a mini-batch stochastic gradient descent (SGD) algorithm. A variety of SGD algorithms have been proposed in the literature and are heavily used in machine learning, see \cite{SG} for a survey.  A generic step in a (stochastic) gradient method is of the form:
\begin{equation}
\label{eq:sgd1}
\theta_{r+1} = \theta_r - \eta_r \nabla_{\theta}\left(J(\theta_k) + \lambda \reg(\theta_k)\right),
\end{equation}
with $\eta_r$ being the \emph{learning rate}. The stochasticity arises in approximating the gradient in \eqref{eq:sgd1} by,
\begin{equation}
\label{eq:sgd2}
\nabla_{\theta} J(\theta_k) \approx \nabla_{\theta}\left(\sum_{y \in \hat{\train}_q}   |\map^{\Delta}(y) - \map_{\theta} (y) |^p \right),
\end{equation}
and analogously for the gradient of the regularization term in \eqref{eq:sgd1}. Here $\hat{\train}_q \subset \train$ refers to the $q$-th batch, with the batches being shuffled randomly. Moreover, the SGD methods are initialized with a starting value $\theta_0 = \bar{\theta} \in \Theta$. A widely used variant of the SGD method is the so-called \emph{ADAM} algorithm \cite{ADAM}. 

For notational simplicity, we denote the (approximate, local) minimum weight vector in \eqref{eq:lf2} as $\theta^{\ast}$ and the underlying deep neural network $\map_{\theta^{\ast}}$ will be our neural network surrogate for the parameters to observable map $\map$ \eqref{eq:ptoob1}. The algorithm for computing this neural network is summarized below,
\begin{algorithm} 
\label{alg:DL} {\bf Deep learning of parameters to observable map}. 
\begin{itemize}
\item [{\bf Inputs}:] Parameterized PDE \eqref{eq:cdep}, Observable \eqref{eq:obsp}, high-resolution numerical method for solving \eqref{eq:cdep} and calculating \eqref{eq:obsp}.
\item [{\bf Goal}:] Find neural network $\map_{\theta^{\ast}}$ for approximating the parameters to observable map $\map$ \eqref{eq:ptoob1}. 
\item [{\bf Step $1$}:] Choose the training set $\train$ and evaluate $\map^{\Delta}(y)$ for all $y \in \train$ by high-resolution CFD simulations. 
\item [{\bf Step $2$}:] For an initial value of the weight vector $\overline{\theta} \in \Theta$, evaluate the neural network $\map_{\overline{\theta}}$ \eqref{eq:ann1}, the loss function \eqref{eq:lf2} and its gradients to initialize the
(stochastic) gradient descent algorithm.
\item [{\bf Step $3$}:] Run a stochastic gradient descent algorithm of form \eqref{eq:sgd1} till an approximate local minimum $\theta^{\ast}$ of \eqref{eq:lf2} is reached. The map $\map^{\ast} = \map_{\theta^{\ast}}$ is the desired neural network approximating the
parameters to observable map $\map$.
\end{itemize}
\end{algorithm}
\subsection{Theory}
\label{sec:theo}
\subsubsection{Approximation with deep neural networks.} 
\label{sec:approx}
 \emph{Universal approximation theorems} \cite{Bar1,Kor1,Cy1} prove that neural networks can approximate any continuous (in fact any measurable) function. However, these universality results are not constructive. More recent papers such as \cite{YAR1,Pet2} provide the following constructive approximation result:
 \begin{theorem}
\label{theo:1} Let $f \in W^{s,p}\left([0,1]^d,\R\right)$ for some $1 \leq p \leq \infty$, such that $\|f\|_{W^{s,p}} \leq 1$, then for every $\epsilon > 0$, there exists a neural network $NN(f)$ of the form \eqref{eq:ann1} with the ReLU activation function, with $\cO\left(1 + \log\left(\frac{1}{\epsilon}\right)\right)$ layers and of size (number of weights) $\cO\left(\epsilon^{-\frac{d}{s}}\left(1 + \log\left(\frac{1}{\epsilon} \right)\right)\right)$ such that
\begin{equation}
\label{eq:approx1}
\|f - NN(f) \|_{L^p} \leq \epsilon, \quad 1 \leq p \leq \infty.
\end{equation}

\end{theorem}
Given the above approximation result, it is essential to investigate regularity of the underlying parameters to observeble map $\map$ in order to estimate the size of neural networks, needed to approximate them. 
\subsubsection{Regularity of the parameters to observable map \eqref{eq:ptoob}}
\label{sec:reg}
For simplicity, we assume that the underlying domain $D_y = D \subset B_R$ for a. e. $y \in Y$ and $B_R$ is a ball of radius $R$. Moreover, we assume that $\psi, g^k \in L^{\infty}$, with $g^k$ denoting the $k$-th derivative of the map $g$ in \eqref{eq:obsp}, for some $k \geq 1$. Then, it is straightforward to obtain the following upper bound,
\begin{equation}
\label{eq:reg1}
\|\map\|_{W^{k,p}(Y)} \leq C \|\U\|_{W^{k,p}(D\times(0,T) \times Y)},
\end{equation}
with the constant $C$ depending only on the initial data, $\Psi$ and $g$. Thus, regularity in the parameters to observable map is predicated on the space-time and parametric regularity of the underlying solution field $\U$.

Unfortunately, it is well-known that one cannot expect such space-time or parametric regularity for solutions of generic convection-diffusion equations \eqref{eq:cde}. In particular, even for the simple case of scalar one-dimensional conservation laws, at best we can show that the solution $\U \in L^{\infty}((0,T);BV(D\times Y))$. Therefore, one can show readily that 
\begin{equation}
\label{eq:reg4}
\|\map\|_{BV(Y)} \leq C.
\end{equation}
It can be rightly argued that deriving regularity estimates on $\map$ in terms of the underlying solution field $\U$ can lead to rather pessimistic results and might overestimate possible cancellations. This is indeed the case, at least for scalar conservation laws with random initial data i.e, \eqref{eq:cdep} with $m=1, \nu = 0$ and ${\bf F}(y,\U) = {\bf F}(\U)$. In this case, we have the following theorem,
\begin{theorem}
\label{theo:21}
Consider the scalar conservation law, i.e parameterized convection-diffusion PDE \eqref{eq:cdep} with $m=1$ and $\nu = 0$, in domain $D = \R^ {d_s}$ and time period $[0,T]$ and assume that the initial data $u_0 \in W^{1,\infty}\left(Y;L^1(D)\right)$ and support of $u_0$ is compact. Assume, furthermore that $\psi \in L^{\infty}(D \times (0,T))$ and $g \in W^{1,\infty} (\R)$. Then, the parameters to observable map $\map$, defined in \eqref{eq:ptoob} satisfies,
\begin{equation}
\label{eq:reg6}
\|\map\|_{W^{1,\infty}(Y)} \leq C.
\end{equation}
for some constant $C$, depending on the initial data $u_0$, $\psi$ and $g$. Moreover, if the numerical solution $u^{\Delta}$ is generated by a monotone numerical method, then the approximate parameters to observable map $\map^{\Delta}$ \eqref{eq:ptoob1} satisfies
\begin{equation}
\label{eq:reg7}
\|\map^{\Delta}\|_{W^{1,\infty}(Y)} \leq C.
\end{equation}
\end{theorem}
The proof of this theorem is a consequence of the $L^1$ stability (contractivity) of the solution operator for scalar conservation laws \cite{DAF1}. In fact, for any $y,y^{\ast} \in Y$, a 
straightforward calculation using the definition \eqref{eq:ptoob} and the Lipschitz regularity of $g$ yields,
\begin{align*}
    |\map(y) - \map(y^{\ast})| &\leq \|\psi\|_{\infty} \|g\|_{{\rm Lip}}\int\limits_0^T \|u(t,.,y) - u(t,.,y^{\ast})\|_{1} dt, \\
    &\leq \|\psi\|_{\infty} \|g\|_{{\rm Lip}}\int\limits_0^T \|u_0(.,y) - u_0(.,y^{\ast})\|_{1} dt, \quad {\rm by~}L^1-{\rm contractivity} \\
    &\leq \|\psi\|_{\infty} \|g\|_{{\rm Lip}} T \|u_0\|_{W^{1,\infty}(Y,L^1(D))} \leq C
\end{align*}
The above proof also makes it clear that bounds such as \eqref{eq:reg6} and \eqref{eq:reg7} will hold for systems of conservation laws and the incompressible Navier-Stokes equations as long as there is some stability of the field with respect to the input parameter vector. On the other hand, we cannot expect any such bounds on the higher parametric derivatives of $\map$, due to the lack of differentiability of the underlying solution field.

Given a bound such as \eqref{eq:reg6} or \eqref{eq:reg7}, we can illustrate the difficulty of approximating the map $\map$ by considering a prototypical problem, namely that of learning the lift and the drag of the RAE2822 airfoil (see section \ref{sec:aero}). For this problem, the underlying parameter space is six-dimensional. Assuming that $s=1$ in theorem \ref{theo:1} and requiring that the approximation error is at most one percent relative error i.e $\epsilon = 10^{-2}$ in \eqref{eq:approx1}, yields a neural network of size $\cO(10^{12})$ tunable parameters and at least $6$ layers. Such a large network is clearly unreasonable as it will be very difficult to train and expensive to evaluate. 
\subsubsection{Trained networks and generalization error}
\label{sec:gen}
It can be argued that approximation theoretic results can severely overestimate errors with \emph{trained networks}. Rather, the relevant measure is the so-called \emph{generalization error }i.e, once a (local) minimum $\theta^{\ast}$ of \eqref{eq:lf2} has been computed, we are interested in the following error,
\begin{equation}
\label{eq:gerr1}
\er_{G} (\theta^{\ast}) := \left(\int\limits_{Y} |\map^{\Delta}(y) - \map_{\theta^{\ast}}(y)|^p d\mu(y) \right)^{\frac{1}{p}}, \quad 1 \leq p \leq \infty.
\end{equation}
In practice, one estimates the generalization error on a so-called \emph{test set} i.e, $\test \subset Y$, with $\#(\test) >> N = \#(\train)$. For instance, $\test$ could consist of i.i.d random points in $Y$, drawn from the underlying distribution $\mu$. In this case, the generalization error \eqref{eq:gerr1} can be estimated by the \emph{prediction error},
 \begin{equation}
\label{eq:perr}
\er_{P} (\theta^{\ast}) : = \left(\frac{1}{\#(\test)}\sum\limits_{y \in \test} |\map^{\Delta}(y) - \map_{\theta^{\ast}} (y) |^p \right)^{\frac{1}{p}}
\end{equation}

It turns out that one can estimate the generalization error \eqref{eq:gerr1} by using tools from machine learning theory \cite{MLbook} as,
\begin{equation}
\label{eq:gerr2}
\er_{G} (\theta^{\ast}) \leq \sqrt{\frac{U}{N}},
\end{equation}
with $N = \#(\train)$ being the number of training samples. The numerator $U$ is typically estimated in terms of the Rademacher complexity or the Vapnik-Chervonenkis (VC) dimension of the network, see \cite{MLbook} and references therein for definitions and estimates.

Unfortunately, such bounds on $U$ are very pessimistic in practice and over estimate the generalization error by several (tens of) orders of magnitude \cite{Arora,REC1, NEYS1}. Recent papers such as \cite{Arora} provide far sharper bounds by estimating the \emph{compression} i.e number of effective parameters to total parameters in a network. 

Nevertheless, even if we make a stringent requirement of $U \sim \cO(1)$, we might end up with unacceptably high prediction errors of $30-100 \%$ for the $\cO(10^2-10^3)$ training samples that we can afford to generate in practice. 

Summarizing, theoretical considerations outlined above indicate the challenges of \emph{finding deep neural networks to accurately learn maps of low regularity in a data poor regime}. We illustrate these difficulties with a simple numerical example. 
\subsubsection{An illustrative numerical experiment.}
In this experiment we consider the parameter space $Y = [0,1]^6$ the following two maps, 
\begin{equation}
    \label{eq:sine1}
    \map^1(y) = \sum\limits_{i=1}^6 \sin(4 \pi y_i), \quad 
    \map^2(y) = \sum\limits_{i=1}^6 \frac{1}{i^3} \sin(4 \pi y_i).
\end{equation}
Both maps are infinitely differentiable but with high amplitudes of the derivatives. 

We define a neural network of the form \eqref{eq:ann1} and with the architecture specified in table \ref{tab:afoilRefNet}. In order to generate the training set, we select the first $128$ Sobol points in $[0,1]^6$ and sample the functions \eqref{eq:sine1} at these points. The training is performed with ADAM algorithm and hyperparameters defined in table \ref{tab:afoilBPNets} (first line). The performance of the neural network is ascertained by choosing the first $8192$ Sobol points in $[0,1]^6$ as the test set $\test$. The results of evaluating the trained networks on the test set is shown in table \ref{tab:sine} where we present the relative prediction error \eqref{eq:perr} as a percentage and also present the standard deviation (on the test set) of the prediction error. Clearly, the errors are unacceptably high, particularly for the unscaled sum of sines $\map^1$ in \eqref{eq:sine1}. Although still high, the errors reduce by an almost order of magnitude for the scaled sum of sines $\map^2$. This is to be expected as the derivatives for this map are smaller. This example illustrates the difficulty to approximating (even very regular) functions, in moderate to high dimensions, by neural networks when the training set is relatively small. 

\begin{table}
    \centering
    \begin{tabular}{|l|l|l|}
    \hline
        Map & Error (Mean) in $\%$ & Error (Std.) in $\%$ \\
    \hline
        $\sum_{i=1}^6 \sin(4\pi x_i)$ & $133.95$ & $417.83$ \\
        \hline
        $\sum_{i=1}^6 \frac{\sin(4\pi x_i)}{i^3}$ & $43.66$ & $41.26$ \\
        \hline
    \end{tabular}
    \caption{Relative percentage Prediction errors \eqref{eq:perr} with $p=2$, for the trained neural networks in the sum of sines experiment \eqref{eq:sine1}.}
    \label{tab:sine}
\end{table}
\subsubsection{Pragmatic choice of network size}
We estimate network size based on the following very deep but easy to prove fact about training neural networks \cite{REC1},
\begin{lemma}
\label{lem:1}
For the map $\map^{\Delta}$ and for the training set $\train$, with $\#(\train) = N$, there exists a weight vector $\hat{\theta} \in \Theta$, and a resulting neural network of the form $\map_{\theta}$, with $\Theta \subset \R^M$ and $M = {\mathcal O} (d + N)$, such that the following holds,
\begin{equation}
\label{eq:lem1}
\map^{\Delta}(z) = \map_{\hat{\theta}} (z), \quad \forall~ z \in \train.
\end{equation}
\end{lemma}
The above lemma implies that there exists a neural network with weight vector $\hat{\theta}$  of size ${\mathcal O}(d+N)$ such that the training error defined by,
\begin{equation}
\label{eq:ert}
\er_{T} (\hat{\theta}) : = \frac{1}{\#(\train)}\sum\limits_{y \in \train} |\map^{\Delta}(y) - \map_{\hat{\theta}} (y) |^p = 0.
\end{equation}
Thus, it is reasonable to expect that a network of size $\cO(d + N)$ can be trained by a gradient descent method to achieve a very low training error.
It is customary to monitor the generalization capacity of a trained neural network by computing a so-called \emph{validation set} $\val \subset Y$ with $\val \cap \train = \emptyset$, and evaluating the so-called validation loss,
\begin{equation}
\label{eq:vlf}
J_{\val} (\theta) : = \frac{1}{\#(\val)}\sum\limits_{y \in \val} \|\map^{\Delta}(y) - \map_{\theta} (y) \|_p^p.
\end{equation}
A low validation loss is observed to correlate with low generalization errors. In order to generate the validation set, one can set aside a small proportion, say $10-20\%$ of the training set as the validation set. 
\subsection{Hyperparameters and Ensemble training.}
Given the theoretical challenges of determining neural networks that can provide low prediction errors in our data poor regime, a suitable choice of the hyperparameters in the training process becomes imperative. To this end, we device a simple yet effective \emph{ensemble training algorithm}. To this end, we consider the set of hyperparameters, listed in table \ref{tab:1}. A few comments regarding the listed hyperparameters are in order. We need to choose the number of layers and width of each layer such that the overall network size, given by \eqref{eq:ns} is $\cO(N+d)$. The exponents $p,q$ in the loss function and regularization terms \eqref{eq:lf2} usually take the values of $1$ or $2$ and the regularization parameter $\lambda$ is required to be small. The starting value $\bar{\theta}$ is chosen randomly. 

For each choice of the hyperparameters, we realize a single sample in the hyperparameter space. Then, the machine learning algorithm \ref{alg:DL} is run with this sample and the resulting loss function minimized. Further details of this procedure are provided in section \ref{sec:imp}. We remark that this procedure has many similarities to the \emph{active learning} procedure proposed recently in \cite{E4}. 
\begin{table}[htbp]
\centering
\begin{tabular}{ll}

1. & Number of hidden layers $(K-1)$ \\
2. & Number of units in $k$-th layer $(d_k)$ \\
3. & Exponent $p$ in the loss function \eqref{eq:lf1}. \\
4. & Exponent $q$ in the regularization term in \eqref{eq:lf2} \\
5. & Value of regularization parameter $\lambda$ in \eqref{eq:lf2} \\
6. & Choice of optimization algorithm (optimizer) -- either standard SGD or ADAM. \\
7. & Initial guess $\bar{\theta}$ in the SGD method \eqref{eq:sgd1}. 
\end{tabular}
\caption{Hyperparameters in the training algorithm}
\label{tab:1}
\end{table}
\section{Uncertainty quantification in CFD with deep learning.}
\label{sec:3}
We will employ the trained deep neural networks that approximate the parameters to observable map $\map^{\Delta}$, to efficiently quantify uncertainty in the underlying map. We are interested in computing the entire measure (probability distribution) of the observable. To this end, we assume that the underlying probability distribution on the input parameters to the problem \eqref{eq:cdep} is given by $\mu \in {\rm Prob}(Y)$. In the context of forward UQ, we are interested in how this initial measure is changed by the parameters to observable map $\map$ (or its high-resolution numerical surrogate $\map^{\Delta}$). Hence, we consider the \emph{push forward} measure $\hat{\mu}^{\Delta} \in {\rm Prob}(\R)$ given by 
\begin{equation}
\label{eq:pf1} 
\hat{\mu}^{\Delta} := \map^{\Delta}\#\mu, \quad \Rightarrow \quad \int_{\R} h(z) d\hat{\mu}^{\Delta}(z) = \int_{Y} h(\map^{\Delta}(y)) d\mu(y), 
\end{equation}
for any $\mu$-measurable function $h:\R \mapsto \R$. It should be emphasized that any statistical moment of interest can be computed by integrating an appropriate test function with respect to this measure $\hat{\mu}^{\Delta}$. For instance, letting $h(w) = w$ in \eqref{eq:pf1} yields the mean of the observable. Similarly, the variance can be computed from the mean and letting $h(w) = w^2$ in \eqref{eq:pf1}. \emph{The task of efficiently computing this probability distribution is significantly harder than just estimating the mean and variance of the underlying map}.

The simplest algorithm for computing this measure $\hat{\mu}^{\Delta}$ is to use the Monte Carlo algorithm (\cite{CAF1} and references therein). It consists of selecting $J$ samples i.e points $\{ y_j \}$ with each $y_j \in Y$ and $1 \leq j \leq J$, that are independent and identically distributed (\emph{iid}). Then, the Monte Carlo approximation of the push-forward measure is given by, 
\begin{equation}
\label{eq:mc}
\hat{\mu}_{mc} = \frac{1}{J}\sum\limits_{j=1}^{J} \delta_{\map^{\Delta}(y_j)} \quad \Rightarrow \quad \int_{\R} h(z) d\hat{\mu}_{mc}(z) = \frac{1}{J}\sum\limits_{j=1}^{J} h\left(\map^{\Delta}(y_j) \right).
\end{equation}
It can be shown that $\hat{\mu}_{mc} \approx \hat{\mu}^{\Delta}$ with an error estimate that scales (inversely) as the square root of the number of samples $J$. Hence, the Monte Carlo algorithm can be very expensive in practice. 
\subsection{Quasi-Monte Carlo (QMC) methods}
\label{sec:qmc}
\subsubsection{Baseline QMC algorithm.}
Quasi-Monte Carlo (QMC) methods are deterministic quadrature rules for computing integrals, \cite{CAF1} and references therein. The key idea underlying QMC is to choose a (deterministic) set of points on the domain of integration, designed to achieve some measure of \emph{equidistribution}. Thus, the QMC algorithm is expected to approximate integrals with higher accuracy than MC methods, whose nodes are randomly chosen \cite{CAF1}. 

For simplicity, we set the underlying measure $\mu$ to be the scaled Lebesgue measure on $Y$. Our aim is to approximate the push forward measure $\hat{\mu}^{\Delta}$ with the QMC algorithm. To this end, we choose a set of points $\Jq = \{ y_j \} \subset Y$ with $1 \leq j \leq J_q = \#(\Jq)$. Then the measure $\hat{\mu}^{\Delta}$ can be approximated by,
\begin{equation}
\label{eq:m1}
\hat{\mu}_{qmc} = \frac{1}{J_q}\sum\limits_{j=1}^{J_d} \delta_{\map^{\Delta}(y_j)} \quad \Rightarrow \quad \int_{\R} h(z) d\hat{\mu}_{qmc}(z) = \frac{1}{J_q}\sum\limits_{j=1}^{J_q} h\left(\map^{\Delta}(y_j) \right),
\end{equation}
for any measurable function $h:Y \to \R$. 

We want to estimate the difference between the measures $\hat{\mu}^{\Delta}$ and $\hat{\mu}_{qmc}$ and need some metric on the space of probability measures, to do so. One such metric is the so-called Wasserstein metric \cite{VIL1}. To define this metric, we need the following, 
\begin{definition}
Given two measures, $\nu,\sigma \in {\rm Prob}(\R)$, a transport plan $\pi\in {\rm Prob}(\R^2)$ is a probability measure on the product space such that the following holds for all measurable functions $F,G$,
\begin{equation}
\label{eq:w1}
\int_{\R \times \R} (F(u) + G(v)) d\pi(u,v) = \int_{\R} F(u) d\nu(u) + \int_{\R} G(v) d\sigma(v).
\end{equation}
The set of transport plans is denoted by $\Pi(\nu,\sigma)$
\end{definition} 
Then the $1$-Wasserstein distance \cite{VIL1} between $\hat{\mu}^{\Delta}$ and $\hat{\mu}_{qmc}$ is defined as,
\begin{equation}
\label{eq:w2}
W_1\left(\hat{\mu}^{\Delta}, \hat{\mu}_{qmc} \right) := \inf\limits_{\pi \in \Pi\left(\hat{\mu}^{\Delta}, \hat{\mu}_{qmc}\right)} \int\limits_{\R \times \R} |u - v| d\pi(u,v)
\end{equation}
It can be proved using the Koksma-Hlawka inequality and some elementary optimal transport techniques that the 
error in the Wasserstein metric behaves as,
\begin{equation}
\label{eq:m3}
W_1\left(\hat{\mu}^{\Delta}, \hat{\mu}_{qmc} \right) \sim V^{\Delta} \Ds(\Jq).
\end{equation}
Here, $V^{\Delta}$ is some measure of \emph{variation} of $\map^{\Delta}$. It can be bounded above by
\begin{equation}
\label{eq:vhk}
V^{\Delta}  \leq C \int\limits_{Y} \left | \frac{\partial^{d}}{\partial y_1 y_2 \cdots y_d} \map^{\Delta}(y) dy \right|,
\end{equation}
the Hardy-Krause variation of the function $\map^{\Delta}$. This upper bound requires some regularity for the underlying map in terms of mixed partial derivatives. However, it is well-known that the bound \eqref{eq:vhk} is a large overestimate of the integration error \cite{CAF1}. 

The $\Ds(\Jq)$ in \eqref{eq:m3} is the so-called \emph{discrepancy} of the point sequence $\Jq$, defined formally in \cite{CAF1} and references therein. It measures how equally is the point sequence $\Jq$ spread out in $Y$. The whole objective in the development of quasi-Monte Carlo (QMC) methods is to design \emph{low discrepancy sequences}. In particular, there exists many popular QMC sequences such as Sobol, Halton or Niederreiter \cite{CAF1} that satisfy the following bound on discrepancy,
\begin{equation}
\label{eq:lds}
  \Ds(\Jq) \sim \frac{(\log(J_q))^{d}}{J_q}.
  \end{equation}
  Thus, the integration error with QMC quadrature based on these rules behaves log-linearly with respect to the number of quadrature points. For simplicity, we assume that the dimension is moderate enough to replace the logarithms in \eqref{eq:m3} by a suitable power leading to 
  \begin{equation}
\label{eq:m2}
W_1\left(\hat{\mu}^{\Delta}, \hat{\mu}_{qmc} \right) \sim \left(\frac{V^{\Delta}}{J_q}\right)^{\alpha},
\end{equation}
for some $1/2 < \alpha \leq 1$. Thus, for moderate dimensions, we see that the QMC algorithm outperforms the baseline MC algorithm. 

Requiring that the Wasserstein distance is of size $\cO(\epsilon)$ for some tolerance $\epsilon > 0$, entails choosing $J_q \sim \frac{V^{\Delta}}{\epsilon^{\frac{1}{\alpha}}}$ and leads to a cost of 
\begin{equation}
\label{eq:cqmc1}
\hat{C}_{qmc} \sim \cost \frac{V^{\Delta}}{\epsilon^{\frac{1}{\alpha}}},
\end{equation}
with $\cost$ being the computational cost of a single CFD forward solve. This can be very expensive, particularly for smaller values of $\alpha$, as the cost of each forward solve is very high.
\subsubsection{Deep learning quasi-Monte Carlo (DLQMC)}
\label{sec:dlqmc}
In order to accelerate the baseline QMC method, we describe an algorithm that combines the QMC method with a deep neural network approximation of the underlying parameters to observable map.
\begin{algorithm}
\label{alg:dlqmc}
{\bf Deep learning Quasi-Monte Carlo (DLQMC).} 
\begin{itemize}
\item [{\bf Inputs:}] Parameterized PDE \eqref{eq:cdep}, Observables \eqref{eq:obsp}, high-resolution numerical method for solving \eqref{eq:cdep} and calculating \eqref{eq:obsp}. 
\item [{\bf Goal:}] Approximate the push-forward measure $\hat{\mu}^{\Delta}$.
\item [{\bf Step 1:}] For $J_L \in \N$, select $\JL = \{y_j\} \subset Y$, with $1 \leq j \leq J_L$, as the first $J_L$ points from a Quasi-Monte Carlo low-discrepancy sequence such as Sobol, Halton etc. 
\item [{\bf Step 2:}] For some $N << J_L$, denote $\train= \{y_j\} \subset \JL$, with $1 \leq j \leq N$ i.e, first $N$ points in $\JL$, and evaluate the map $h(\map^{\Delta}(y_j))$, for all $y_j \in \Jd$ with the high-resolution CFD simulation. As a matter of fact, any consecutive set of $N$ QMC points can serve as the training set. 
\item [{\bf Step 3:}]  With $\train$ as the training set,  compute an optimal neural network $\map^{\ast}$ for the parameters to
observable map $\map^{\Delta}$, by using algorithm \ref{alg:DL}.
\item [{\bf Step 4:}]  Approximate the measure $\hat{\mu}^{\Delta}$ by
\begin{equation}
\label{eq:dlqmc1}
\hat{\mu}^{\ast}_{qmc} = \frac{1}{J_L}\sum\limits_{j=1}^{J_L} \delta_{\map^{\ast}(y_j)},
\end{equation}
\end{itemize}
\end{algorithm}
The complexity of the DLQMC algorithm is analyzed in the following theorem. 
\begin{theorem}
\label{thrm:4}
For any given tolerance $\epsilon > 0$ and under the assumption that the generalization error \eqref{eq:gerr1} is of size $\cO(\epsilon)$ and the baseline QMC estimator \eqref{eq:m1} follows the error estimate \eqref{eq:m2}, the speedup $\Sigma_{dlqmc}$, defined as the ratio of the cost of baseline QMC algorithm and the cost of the DLQMC algorithm \ref{alg:dlqmc} for approximating the push-forward measure $\hat{\mu}^{\Delta}$ to an error of size ${\mathcal O} (\epsilon)$ in the $1$-Wasserstein metric, with the DLQMC algorithm \ref{alg:dlqmc}, satisfies
\begin{equation}
\label{eq:sp7}
\frac{1}{\Sigma_{dlqmc}} \sim \frac{N}{J_q} +\frac{\cost_{\ast}}{\cost} \frac{V^{\ast}}{V^{\Delta}}.
\end{equation}
Here, $\cost, \cost_{\ast}$ are the computational cost of computing $\map^{\Delta}(y)$ (high resolution CFD simulation) and $\map^{\ast}(y)$ (deep neural network) for any $y \in Y$ and $V^{\ast} = V(\map^{\ast})$ is an estimate on the variation that arises in a QMC estimate such as \eqref{eq:m2}. 
\end{theorem}
\begin{proof}
We define the measure $\hat{\mu}^{\ast} = \map^{\ast}\# \mu \in {\rm Prob}(\R)$ and estimate,
$$
W_1\left(\hat{\mu}^{\Delta},\hat{\mu}^{\ast}_{qmc} \right) \leq W_1\left(\hat{\mu}^{\Delta},\hat{\mu}^{\ast} \right) + W_1\left(\hat{\mu}^{\ast},\hat{\mu}^{\ast}_{qmc} \right).
$$
We claim that assuming that the generalization error $\er_G \sim \epsilon$ implies that $W_1\left(\hat{\mu}^{\Delta},\hat{\mu}^{\ast} \right) \sim \epsilon$.  To see this, we define a transport plan $\pi^{\ast} \in {\rm Prob}(\R^2)$ by $\pi^{\ast} = \hat{\mu}^{\Delta} \otimes \hat{\mu}^{\ast}$. Then, 
\begin{align*}
\int_{\R^2} |u-v| d\pi^{\ast}(u,v) &= \int_{\R^2} |u-v| d(\map^{\Delta}\#\mu \times \map^{\ast}\#\mu)(u,v), \\
&= \int_{Y} |\map^{\Delta}(y) - \map^{\ast}(y)| d\mu(y) := \er_G \\
&\sim \epsilon.
\end{align*} 
This provides an upper bound on the Wasserstein distance \eqref{eq:w2} and proves the claim.  Similarly, we use an estimate of the type \eqref{eq:m2} to obtain $W_1\left(\hat{\mu}^{\ast},\hat{\mu}^{\ast}_{qmc} \right) \sim \left(\frac{V^{\ast}}{J_L}\right)^{\alpha}$. Requiring this error to be $\cO(\epsilon)$ yields $J_L \sim \frac{V^{\ast}}{\epsilon^{\frac{1}{\alpha}}}$ and leads to the following estimate on the total cost of the DLQMC algorithm,
\begin{equation}
\label{eq:cqmc2}
\hat{C}_{dlqmc} \sim \cost N + \cost_{\ast} \frac{V^{\ast}}{\epsilon^{\frac{1}{\alpha}}},
\end{equation}
with $\cost,\cost_{\ast}$ being the computational cost of a single CFD forward solve and evaluation of the deep neural network $\map^{\ast}$, respectively.  

Dividing \eqref{eq:cqmc2} with \eqref{eq:cqmc1} and using $J_q \sim \frac{V^{\Delta}}{\epsilon^{\frac{1}{\alpha}}}$ leads to the speedup estimate \eqref{eq:sp7}.

\end{proof}
\begin{remark}
    In particular, we know that the cost $\cost = {\mathcal O}\left(\Delta ^{-\frac{1}{d+1}}\right)$ is very high for a small mesh size $\Delta$. On the other hand the cost $\cost_{\ast} = \cO(M)$ with $M$ being the number of neurons \eqref{eq:ns} in the network. This cost is expected to be much lower than $\cost$. Moreover, it is reasonable to expect that $V^{\ast} \approx V^{\Delta}$. As long as we can guarantee that $N << J_q$, it follows from \eqref{eq:sp7} that the speedup $\Sigma_{dlqmc}$, which measures the gain of using the deep learning based DLQMC algorithm \ref{alg:dlqmc} over the baseline Quasi-Monte Carlo method, can be quite substantial.
\end{remark} 
\section{Implementation of ensemble training}
\label{sec:imp}
In the following numerical experiments, we use fully connected neural networks with ReLU activation functions. Essential hyperparameters are listed in table \ref{tab:1}. For each experiment, we use a reference network architecture i.e number of hidden layers and width per layer, such that the total network size is $\cO(d + N)$. This ensures consistency with the prescriptions of lemma \ref{lem:1}. 

We use two choices of the exponent $p$ in the loss function \eqref{eq:lf1} i.e either $p=1$ or $p=2$, denoted as mean absolute error (MAE) or mean square error (MSE), respectively. Similarly the exponent $q$ in the regularization term \eqref{eq:lf2} is either $q=1$ or $q=2$. In order to keep the ensemble size tractable, we chose four different values of the regularization parameter $\lambda$ in \eqref{eq:lf2} i.e, $\lambda = 7.8\times10^{-5}, 7.8\times10^{-6}, 7.8\times10^{-7},0$, corresponding to different orders of magnitude for the regularization parameter. For the optimizer of the loss function, we use either a standard SGD algorithm or ADAM. Both algorithms are used in the \emph{full batch} mode. This is reasonable as the number of training samples are rather low in our case.

A key hyperparameter is the starting value $\bar{\theta}$ in the SGD or ADAM algorithm \eqref{eq:sgd1}. We remark that the loss function \eqref{eq:lf2} is non-convex and we can only ensure that the gradient descent algorithms converge to a local minimum. Therefore, the choice of the initial guess is quite essential as different initial starting values might lead to convergence to local minima with very different loss profiles and generalization properties. Hence, it is customary in machine learning to \emph{retrain} i.e, start with not one, but several starting values and run the SGD or ADAM algorithm in parallel. Then, one has to select one of the resulting trained networks. The ideal choice would be to select the network with the best generalization error. However, this quantity is not accessible with the data at hand. Hence, we rely on the following surrogates for predicting the best generalization error,
\begin{itemize}
\item [1.] {\bf Best trained network.} (\emph{train}): We choose the network that leads to the lowest training error \eqref{eq:ert}, once training has been terminated. 
\item [2.] {\bf Best validated network.} (\emph{val}): We choose the network that leads to the lowest validation loss \eqref{eq:vlf}, once training has been terminated. One disadvantage of this approach is to sacrifice some training data for the validation set. 
\item [3.]{\bf Best trained network wrt mean.} (\emph{mean-train}): We choose the network that has minimized the error in the mean of the training set i.e,
\begin{equation}
\label{eq:ermean}
\er_{mean} := \left| \frac{1}{N}\sum\limits_{j=1}^N \map^{\Delta}(y_j) -  \frac{1}{N}\sum\limits_{j=1}^N \map^{\ast}(y_j)\right|, \quad \forall y_j \in \train.
\end{equation}
\item [4.] {\bf Best trained network wrt Wasserstein.} (\emph{wass-train}): We choose the network that has minimized the error in the Wasserstein metric with respect to the training set $\train$:
\begin{equation}
\label{eq:erwass}
\er_{wass} := W_1\left( \frac{1}{N}\sum_{j=1}^N \delta_{\map^{\Delta}(y_j)},   \frac{1}{N}\sum_{j=1}^N \delta_{\map^{\ast}(y_j)} \right), \quad \forall y_j \in \train.
\end{equation}
\end{itemize}
We note that Wasserstein metric for measures on $\R$ can be very efficiently computed with the Hungarian algorithm \cite{HUN1}. 

Another hyperparameter (property) is whether we choose randomly distributed points (Monte Carlo) or more uniformly distributed points, such as quasi-Monte Carlo (QMC) quadrature points, to select the training set $\train$. We use both sets of points in the numerical experiments below in order to ascertain which choice is to be preferred in practice. Other important hyperparameters in machine learning are the \emph{learning rate} in \eqref{eq:sgd1} and the number of training epochs. Unless otherwise stated, we keep them fixed to the values $\eta_r = = 0.01, \forall r$, for both the SGD and ADAM algorithms and to $500000$ epochs, respectively. 
\subsection{Code}
\label{sec:code}
All the following numerical experiments use finite volume schemes as the underlying CFD solver. The machine learning and ensemble training runs are performed by a collection of Python scripts and Jupyter notebooks, utilizing Keras~\cite{keras} and Tensorflow~\cite{tensorflow} for machine learning and deep neural networks. The ensemble runs over hyperparameters are parallelized as standalone processes, where each process trains the network for one choice of hyperparameters. Jupyter notebooks for all the numerical experiments can be downloaded from \url{https://github.com/kjetil-lye/learning_airfoils}, under the MIT license. A lot of care is taken to ensure reproducability of the numerical results. All plots are labelled with the git commit SHA code, that produced the plot, in the upper left corner of the plot.
\begin{figure}[htbp]
\subfigure[Airfoil shape]{\includegraphics[width=0.3\textwidth]{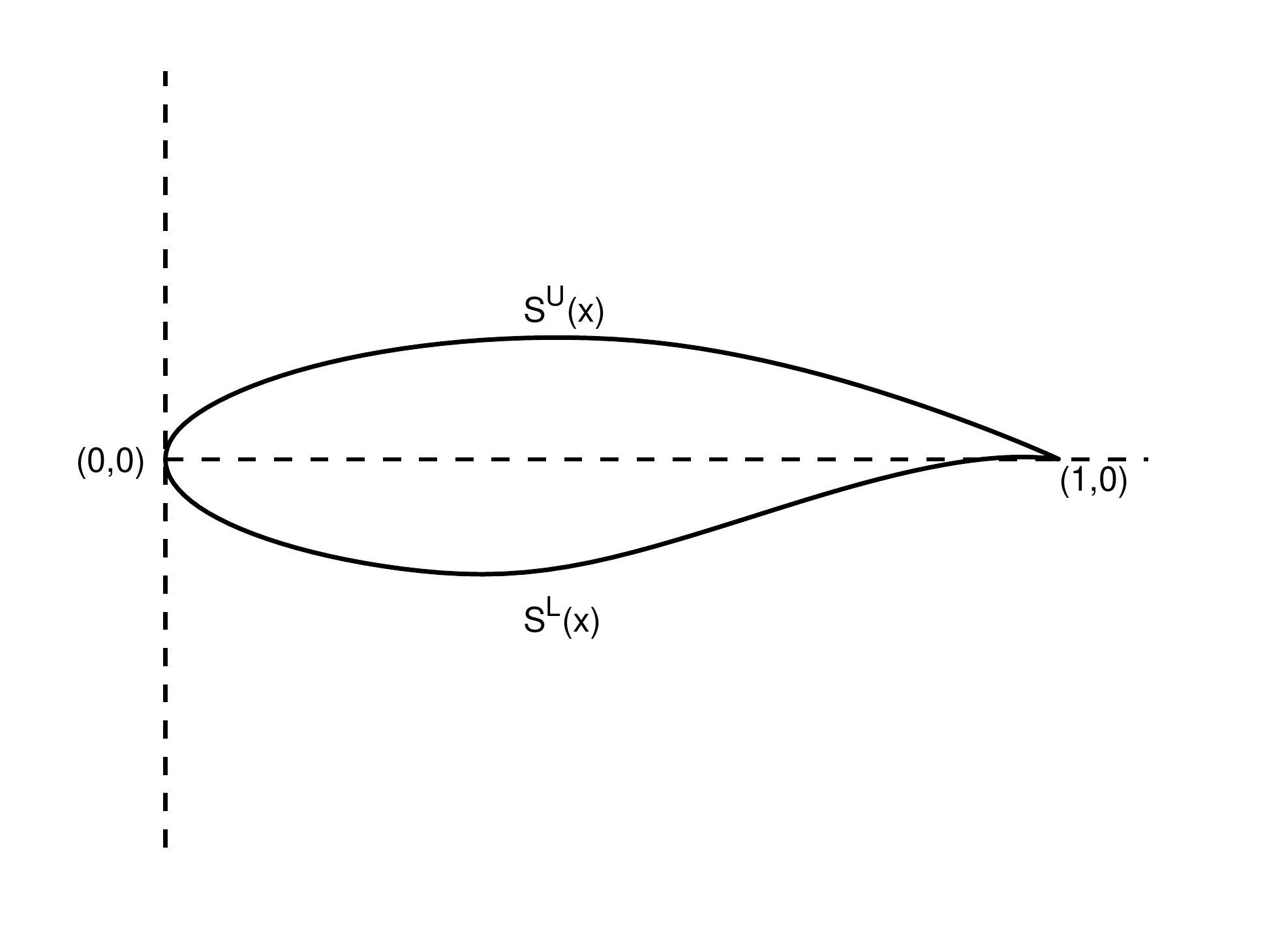}}
\subfigure[Primary mesh]{\includegraphics[width=0.3\textwidth]{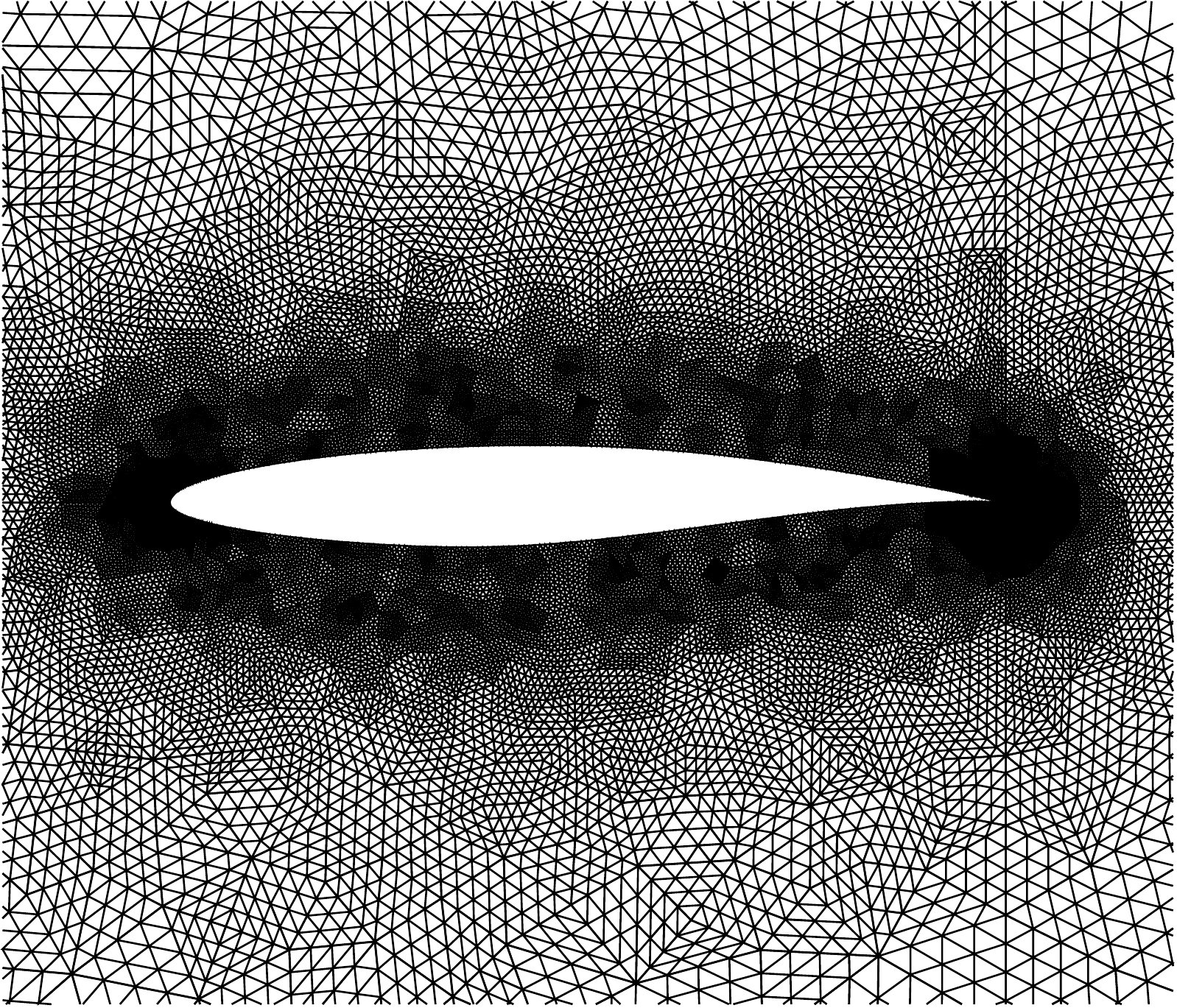}}
\subfigure[Dual mesh]{\includegraphics[width=0.3\textwidth]{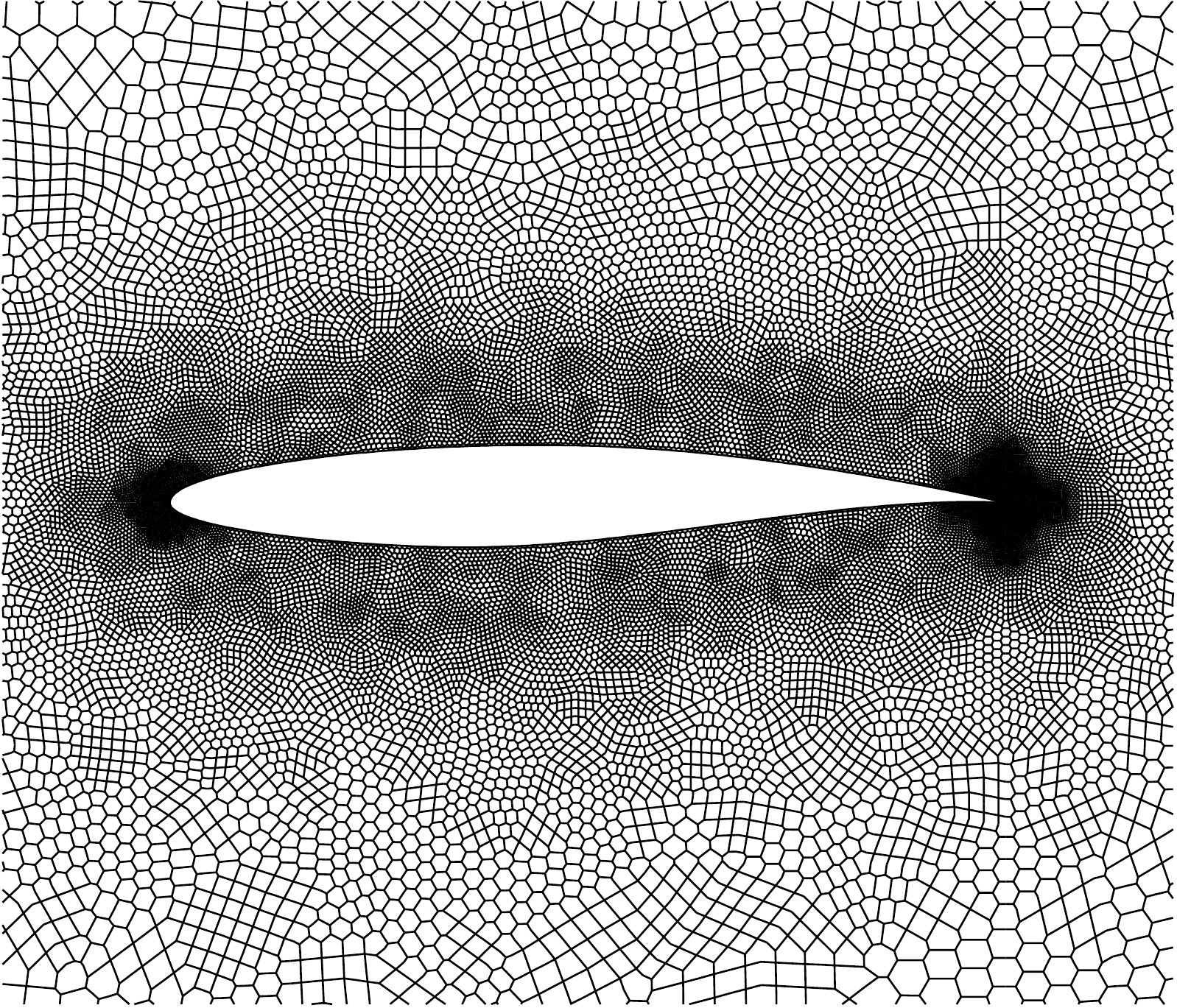}}
\caption{The shape of the RAE airfoil and the underlying primary and voronoi dual meshes.}
\label{fig:rae_mesh}
\end{figure}
\section{Numerical results}
\label{sec:num}
In the following numerical experiments, we consider the compressible Euler equations of gas dynamics as the underlying PDE \eqref{eq:cde}. In two dimensions, these equations are 
\begin{equation}
\label{eq:euler2D}
\U_t + \F^x(\U)_x + \F^y(\U)_y = 0,  \quad
\U = \begin{pmatrix}
\rho \\
\rho u \\
\rho v\\
E
\end{pmatrix}, \quad \F^x(\U) = \begin{pmatrix}
\rho u \\
\rho v^2 + p \\
\rho uv \\
(E + p) u
\end{pmatrix}, \quad \F^y(\U) = \begin{pmatrix}
\rho v \\
\rho uv \\
\rho v^2 + p \\
(E + p) v
\end{pmatrix}
\end{equation}
where $\rho, \vel = (u,v)$ and $p$ denote the fluid density, velocity and pressure, respectively. The quantity $E$ represents the total energy per unit volume
\[
E = \frac{1}{2} \rho |\vel|^2 + \frac{p}{\gamma -1},
\]
where $\gamma=c_p/c_v$ is the ratio of specific heats, chosen as $\gamma=1.4$ for our simulations. Additional important variables associated with the flow include the speed of sound $a = \sqrt{\gamma p/\rho}$, the Mach number $M=|\vel|/a$ and the fluid temperature $T$ evaluated using the ideal gas law $p = \rho R T$, where $R$ is the ideal gas constant.
\subsection{Flow past an airfoil}
\label{sec:aero}
\subsubsection{Problem description.}
We consider flow past an RAE 2822 airfoil, whose shape is shown in Figure \ref{fig:rae_mesh} (Left). The flow past this airfoil is a benchmark problem for UQ in aerodynamics \cite{UMRIDA}. The upper and lower surface of the airfoil are denoted at $(x,S^U(x))$ and $(x,S^L(x))$, with $x \in [0,1]$. We consider the following parametrized (free-stream) initial conditions and perturbed airfoil shapes
\begin{equation}
\label{eq:pert}
\begin{aligned}
T^\infty(y) &= 1 + \eps_1 G_1(y), \quad M^\infty(y) = 0.729(1 + \eps_1 G_2(y)), \quad \alpha(y)=[2.31(1+G_3(y))]^\circ,\\  
p^\infty(y) &= 1 + \eps_1 G_4(y), \quad \hat{S}^U(x;y) = S^U(x)(1+\eps_2G_5(y)), \quad \hat{S}^L(x;y) = S^L(x)(1+\eps_2G_6(y)),\\
\end{aligned}
\end{equation}
where $\alpha$ is the angle of attack, the gas constant $R=1$, $\eps_1 = 0.1$, $\eps_2 = 0.2$, $y \in Y = [0,1]^6$ is a parameter and $G_k(y) = 2y_k-1$ for $k=1,...,6$. 

The total drag and lift experienced on the airfoil surface $\hat{S}(y) = \hat{S}^L(y) \bigcup \hat{S}^U(y)$ are chosen as the observables for this experiment. These are expressed as
\begin{equation}
\text{Lift} = \EuScript{L}_1(y) = \frac{1}{\EuScript{K}(y)}\int_{\hat{S}(y)} p(y) (\n(y) \cdot \psi_l(y)) ds, \quad \text{Drag} = \EuScript{L}_2(y) = \frac{1}{\EuScript{K}(y)}\int_{\hat{S}(y)} p(y) (\n(y) \cdot \psi_d(y)) ds,
\end{equation}
where $\EuScript{K}(y) = \rho^\infty(y) |\vel^\infty(y)|^2/2$ is the free-stream kinetic energy, while $\psi_l(y) = [-\sin(\alpha(y)),\ \cos(\alpha(y))]$ and $\psi_d(y) = [\cos(\alpha(y)), \ \sin(\alpha(y))]$.
\begin{table}[htbp]
\centering
\begin{tabular}{|c|c|c|}
\hline
\textbf{Layer}          & \textbf{\begin{tabular}[c]{@{}c@{}}Width \\ (Number of Neurons)\end{tabular}} & \textbf{Number of parameters} \\ \hline
Hidden Layer 1 & 12                                                                           & 84                           \\ \hline
Hidden Layer 2 & 12                                                                           & 156                         \\ \hline
Hidden Layer 3 & 10                                                                            & 130                          \\ \hline
Hidden Layer 4 & 12                                                                            & 132                           \\ \hline
Hidden Layer 5 & 10                                                                            & 130                          \\ \hline
Hidden Layer 6 & 12                                                                           & 156                         \\ \hline
Hidden Layer 7 & 10                                                                            & 130                          \\ \hline
Hidden Layer 8 & 10                                                                            & 110                          \\ \hline
Hidden Layer 9 & 12                                                                            & 132                           \\ \hline
Output Layer   & 1                                                                             & 13                            \\ \hline
                         &                                                                               & \textbf{1149}                  \\ \hline
\end{tabular}
\caption{Reference neural network architecture for the flow past airfoil problem.}
\label{tab:afoilRefNet}
\end{table}

The domain is discretized using triangles to form the primary mesh, while the control volumes are the voronoi dual cells formed by joining the circumcenters of the triangles to the face mid-points. A zoomed view of the primary and dual meshes are shown in Figure \ref{fig:rae_mesh}. The solutions are approximated using a second-order finite volume scheme, with the mid-point rule used as the quadrature to approximate the cell-boundary integrals. The numerical flux is chosen as the kinetic energy preserving entropy stable scheme \cite{RCFM}, with the solutions reconstructed at the cell-interfaces using a linear reconstruction with the minmod limiter . For details on flux expression and the reconstruction procedure, we refer interested readers to \cite{RCFM}. The scheme is integrated using the SSP-RK3 time integrator. The contour plots for the Mach number with two different realizations of the parameter $y \in [0,1]^6$ are shown in Figure \ref{fig:sample_rae}. It should be emphasized that the stochastic perturbations in \eqref{eq:pert} are rather strong, representing on an average a $10-20$ percent standard deviation (over mean). Hence, there is a significant statistical spread in the resulting solutions. As shown in figure \ref{fig:sample_rae}, one of the configurations lead to a transsonic shock around the airfoil while another leads to a shock-free flow around the airfoil. 
\begin{figure}
\subfigure[Sample 1]{\includegraphics[width=0.45\textwidth]{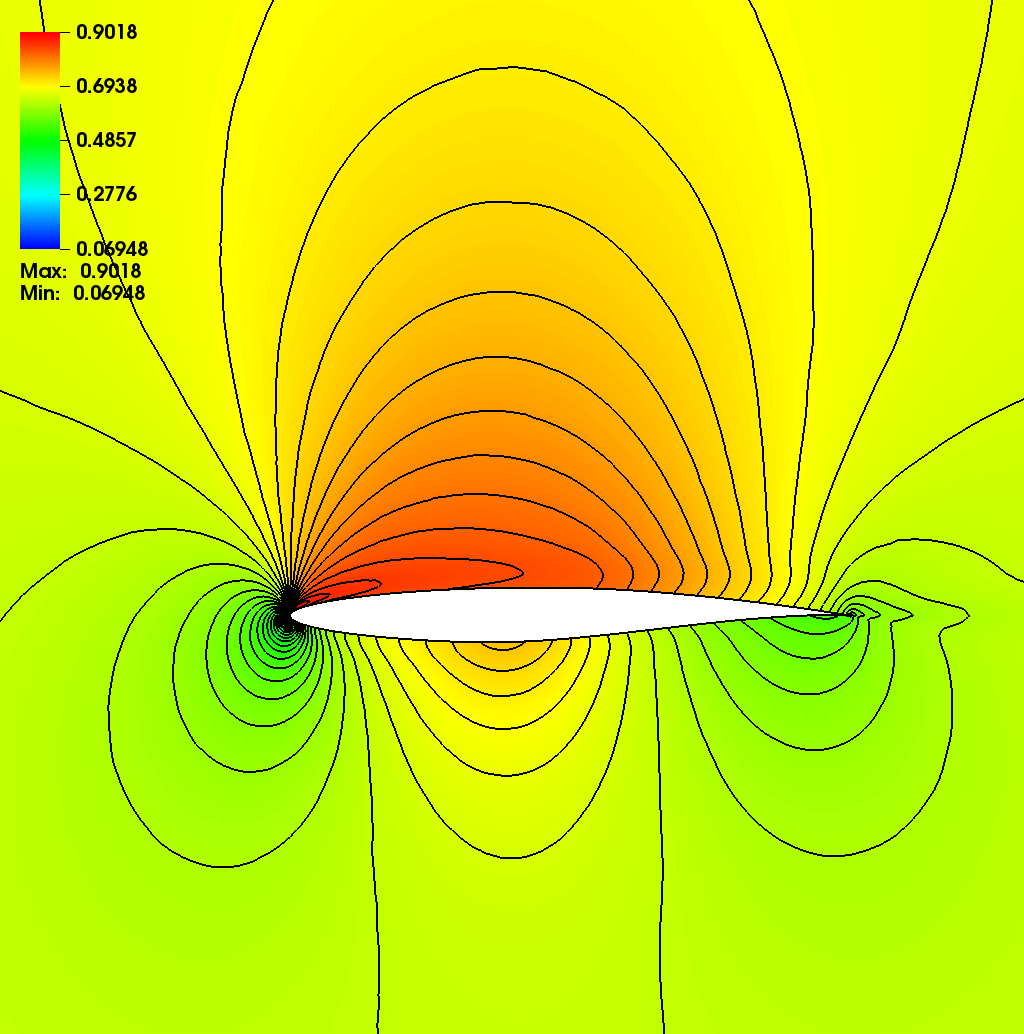}}
\subfigure[Sample 2]{\includegraphics[width=0.45\textwidth]{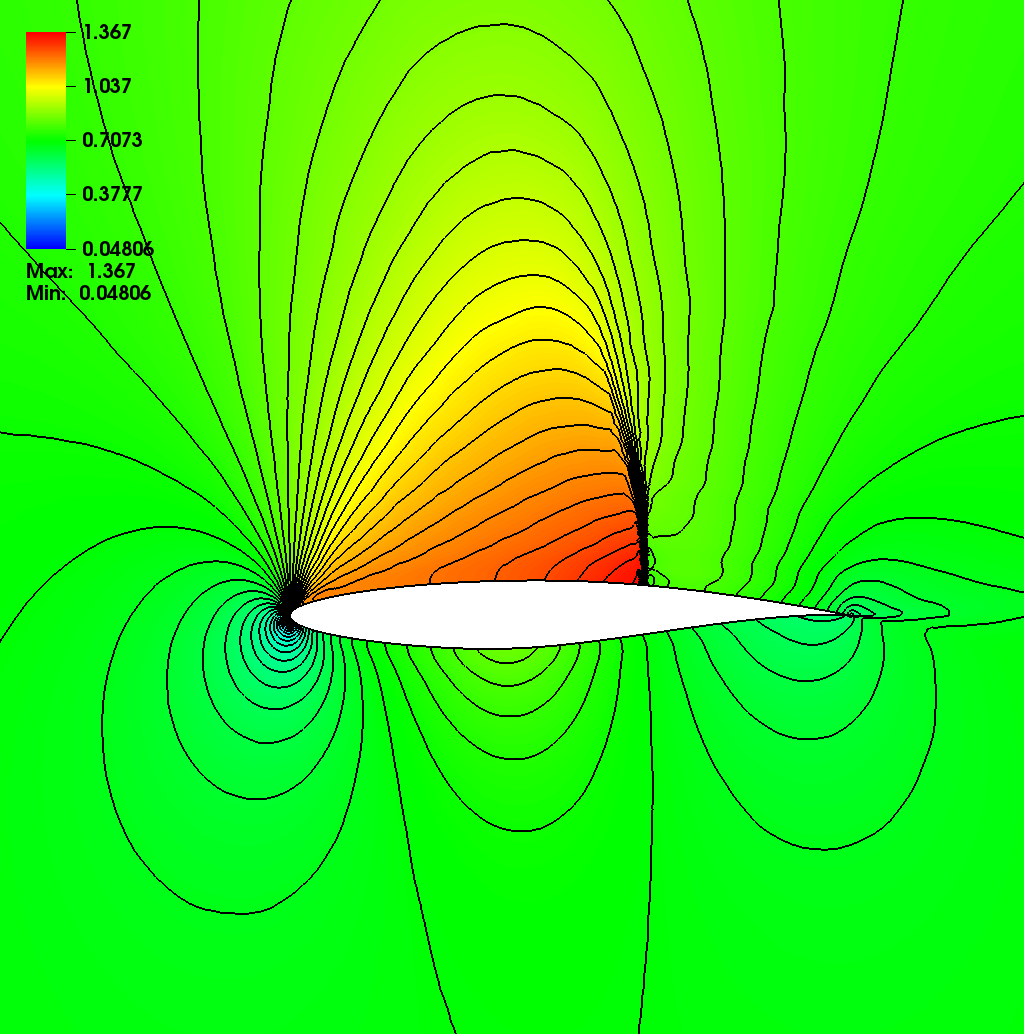}}
\caption{Numerical solutions (Mach number) for the flow past an RAE airfoil, with two different realizations of the parameter y.}
\label{fig:sample_rae}
\end{figure}
\subsubsection{Generation of training data.}
In order to generate the training, validation and test sets, we denote $\EuScript{Q}^{sob}$ as the first $1001$ Sobol points on the domain $[0,1]^6$. For each point $y_j \in\EuScript{Q}^{sob} $, the maps or (rather their numerical surrogate) $\map_{1,2}^{\Delta}(y_j)$ is then generated from high-resolution numerical approximations $\U^\Delta(y)$  and the corresponding lift and drag are calculated by a numerical quadrature . This forms the test set $\test$. For this numerical experiment, we take the first $128$ Sobol points to generate the training set $\train$ and the next $128$ points to generate the validation set $\val$.
\subsubsection{Results of the Ensemble training procedure.}
To set up the ensemble training procedure of the previous section, we select a fully connected network, with number of layers and size of each layer listed in table \ref{tab:afoilRefNet}. Our network has $1149$ tunable parameters and this choice is clearly consistent with the prescriptions of Lemma \ref{lem:1}. Varying the hyperparameters as described in section \ref{sec:imp} led to a total of $114$ configurations and each was \emph{retrained} $5$ times with different random initial choices of the weights and the biases. Thus, we trained a total of $570$ networks in parallel. 

For each configuration, we select the retraining that minimizes the set selection criteria for the configuration. Once the network with best retraining is selected, it is evaluated on the whole test set $\test =  \EuScript{Q}^{sob}$ and the percentage relative $L^2$ prediction error \eqref{eq:perr} computed. Histograms, describing the probability distribution over all hyperparameter samples for the lift and the drag are shown in figure \ref{fig:afoilHistErr}. From this ensemble, we choose the hyperparameter configuration that led to the smallest prediction error and list them in table \ref{tab:afoilBPNets}. As seen from this table, there is a difference in the best performing network corresponding to lift and to drag. However, both networks use ADAM as the optimizer and regularize with the $L^2$-norm of the weights in the loss function \eqref{eq:lf2}.
\begin{table}[htbp]
\centering
\begin{tabular}{|l|l|l|l|l|l|l|}
\hline
Obs & Opt & Loss & $L^1$-reg & $L^2$-reg & Selection & \emph{BP}. Err mean (std)  \\
\hline
Lift & ADAM & MSE & $0.0$ & $7.8\times 10^{-6}$ & wass-train & $0.786$ ($0.010$) \\
\hline
Drag & ADAM & MAE & $0.0$ & $7.8 \times 10^{-6}$ & mean-train & $1.847$ ($0.022$)  \\
\hline 
\end{tabular}
\caption{The hyperparameter configurations that correspond to the best performing network (one with least mean prediction error) for the two observables of the flow past airfoils.}
\label{tab:afoilBPNets}
\end{table}
\begin{figure}[htbp]
\subfigure[Lift]{\includegraphics[width=0.5\textwidth]{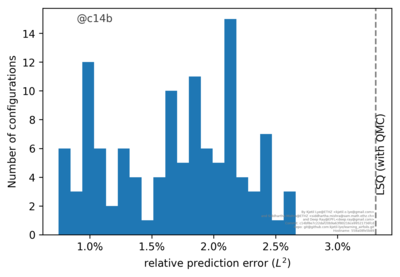}} 
\subfigure[Drag]{\includegraphics[width=0.5\textwidth]{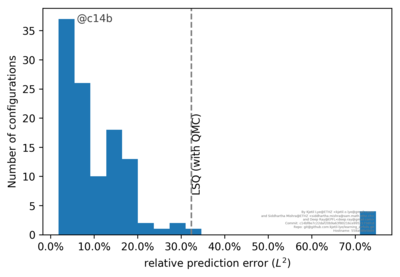}}
\caption{Histograms depicting distribution of the percentage relative mean $L^2$ prediction error \eqref{eq:perr} (X-axis) over number of hyperparameter configurations (samples, Y-axis)  for the lift and the drag in the flow past airfoil problem. Note that the range of $X$-axis is different in the plots.}
\label{fig:afoilHistErr}
\end{figure}
\begin{figure}[htbp]
\subfigure[Lift, Best Performing]{\includegraphics[width=0.45\textwidth]{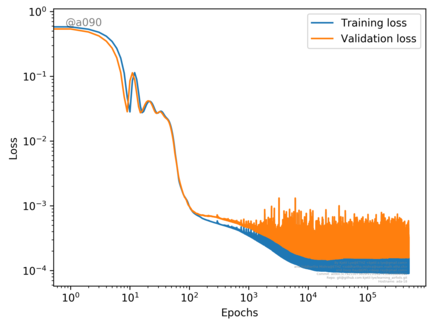}} 
\subfigure[Drag, Best Performing]{\includegraphics[width=0.45\textwidth]{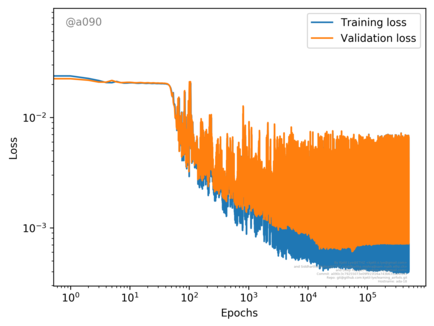}}  \\
\subfigure[Lift, Best Performing]{\includegraphics[width=0.45\textwidth]{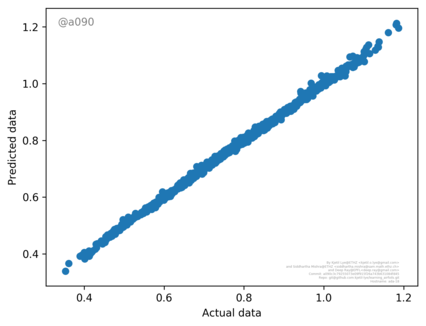}} 
\subfigure[Drag, Best Performing]{\includegraphics[width=0.45\textwidth]{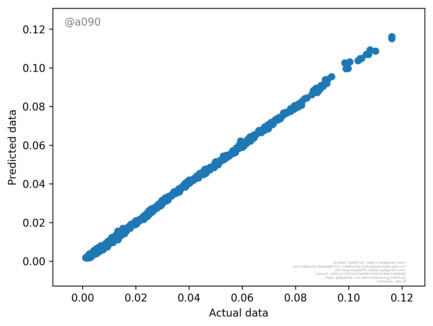}} 
\caption{Training loss \eqref{eq:ert} and Validation loss \eqref{eq:vlf} (Top) and prediction errors (Ground truth vs. prediction) (Bottom) for the best performing (table \ref{tab:afoilBPNets}) neural networks for the flow past airfoils problem as a function of the number of
epochs.}
\label{fig:afoiltrpr}
\end{figure}
We plot the training (and validation) loss, for these best performing networks with respect to the number of epochs of the ADAM algorithm in figure \ref{fig:afoiltrpr} (Top) and observe a three orders (for lift) and two orders (for drag) reduction in the underlying loss function during training.

From figure \ref{fig:afoiltrpr} (Bottom) and table \ref{tab:afoilBPNets}, we note that the prediction errors for the best performing networks are very small i.e, less than $1\%$ relative error for the lift and less than $2\%$ relative error for the drag, even for the relatively small number ($128$) of training samples. This is orders of magnitude less than the sum of sines regression problem (table \ref{tab:sine}). These low prediction errors are particularly impressive when viewed in term of the theory presented in section \ref{sec:theo}. A crude upper bound on the generalization error is given by $\sqrt{\frac{M}{N}}$, with $M,N$ being the number of parameters in the network and number of training samples respectively. Setting these numbers for this problem yields a relative error of approximately $300 \%$. However, our prediction errors are $100-200$ times (two orders of magnitude) less, demonstrating significant \emph{compression} in the trained networks \cite{Arora}. These results are further contextualized, when considered with reference to the computational costs, shown in table \ref{tab:afoilcost}. In this table, we present the costs of generating a single sample, the average training time (over a subset of hyperparameter configurations) and the time to evaluate a trained network. Each sample is generated on the EULER high-performance cluster of ETH-Zurich. In this case, the training time is a very small fraction of the time to generate a single sample, let alone the whole training set whereas the computational cost of evaluating the neural networks is almost negligible as it is nine orders of magnitude less than evaluating a single training sample. 
\begin{table}[]
\centering
\begin{tabular}{|c|c|}
\hline
\textbf{}           & Time (in secs) \\
\hline
Sample generation & $24000$ \\
\hline
Training (Lift) & $700$\\
\hline
Evaluation ($\map_1$) & $9\times 10^{-6}$ \\ 
\hline
Training (Drag) & $840$\\
\hline
Evaluation ($\map_2$) & $10^{-5}$\\
\hline
\end{tabular}
\caption{Computational times (cost) of (single) sample generation, network training and (single) network evaluation, for the flow past airfoils problem. Each sample was generated using a parallelized finite-volume solver on 16 Intel(R) Xeon(R) Gold 5118 @ 2.30GHz processor cores. The entire run took $1500$ secs on the cluster and translates into a total wall clock time of $1500 \times 16 = 24000$ secs. The training and evaluation of the neural networks were all performed on a Intel(R) Core(TM) i7-8700K CPU @ 3.70GHz machine. The training and evaluation times are approximations of the average runtimes over all hyperparameter configurations.}
\label{tab:afoilcost}
\end{table}
A priori, it is unclear if the underlying problem has some special features that makes the compression, and hence high accuracy possible. In particular, learning an affine map is relatively straightforward for a neural network. To test this possibility, we approximate the underlying parameters to observable map, for both lift and drag, with classical \emph{linear least squares algorithm} i.e, $\map^{\text{lsq}}_j:\R^{6}\to \R$ be affine map defined by
\begin{equation}
\label{eq:llsq}
\map^{\ast,\text{lsq}}_j=\arg \min \left\{\sum_{y\in\train} (f(y)-\map_j(y))^2\mid f:\R^6\to \R\text{ is affine}\right\},\quad j=1,2.
\end{equation}
The resulting convex optimization problem is solved by the standard gradient descent. The errors with the linear least squares algorithm are shown in figure \ref{fig:afoilHistErr}. From this figure we can see that the prediction error with the linear least squares is approximately $4 \%$ for lift and $36 \%$ for drag. Hence, the best performing neural networks are approximately $5$ times more efficient for lift and $20$ times more efficient for drag, than the linear least squares algorithm. The gain over linear least squares implicitly measures how \emph{nonlinear} the underlying map is. Clearly in this example, the drag is more nonlinear than the lift and thus more difficult to approximate.
\subsubsection{Network sensitivity to hyperparameters.}
We summarize the results of the sensitivity of network performance to hyperparameters below,
\begin{itemize}
\item \emph{Overall sensitivity.} The overall sensitivity to hyperparameters is depicted as histograms, representing the distribution, over samples (hyperparameter configurations), shown in
figure \ref{fig:afoilHistErr}. As seen from the figure, the prediction errors with different hyperparameters are spread over a significant range for each observable. There are a few outliers which perform very poorly, particularly for the drag. On the other hand,  a large number of  hyperparameter configurations concentrate around the best performing network. Moreover all (most) of the hyperparameter choices led to better prediction errors for the lift (drag), than the linear least squares algorithm.  
\item \emph{Choice of optimizer.} The difference between ADAM and the standard SGD algorithm is shown in figure \ref{fig:afoilOpt}. As in the Sod shock tube problem, ADAM is clearly superior to the standard SGD algorithm. Moreover, there are a few outliers with the ADAM algorithm that perform very poorly for the drag. We call these outliers as the \emph{bad set of ADAM} and they correspond to the configurations with MSE loss function and $L^1$ regularization, or $L^2$ regularization, with a regularization parameter $\lambda > 10^{-5}$, or MAE loss function without any regularization. The difference between the bad set of ADAM and its complement i.e, the good set of ADAM is depicted in figure \ref{fig:afoilAdam}. 
\item \emph{Choice of loss function.} The difference between $L^1$ mean absolute error (MAE) and $L^2$ root mean square error (MSE) i.e, exponent $p=1$ or $p=2$ in the loss function \eqref{eq:lf1} is plotted in the form of a histogram over the error distributions in figure \ref{fig:afoillf}. The difference in performance, with respect to choice of loss function is minor. However, there is an outlier, based on MSE, for the drag.
\item \emph{Choice of type of regularization.} The difference between choosing an $L^1$ or $L^2$ regularization term in \eqref{eq:lf2} is shown in figure \ref{fig:afoilreg}. Again, the differences are minor but the $L^2$ regularization seems to be slightly better, particularly for the drag. 
\item \emph{Value of regularization parameter.} The variation of network performance with respect to the value of the regularization parameter $\lambda$ is shown in figure \ref{fig:afoilregval}. We observe from this figure that a little amount of regularization works better than no regularization. However, the variation in prediction errors wrt respect to the non-zero values of $\lambda$ is minor. 
\item \emph{Choice of selection criteria.} The distribution of prediction errors with respect to the choice of selection criteria for retraining, namely \emph{train,val,mean-train,wass-train} is shown in figure \ref{fig:afoilselect}. There are very minor differences between \emph{mean-train} and \emph{wass-train}. On the other hand, both these criteria are clearly superior to the other two selection criteria. 
\item \emph{Sensitivity to retrainings.} The variation of performance with respect to retrainings i.e, different random initial starting values for ADAM, is provided in table \ref{tab:afoilretrain}. In this table, we list the minimum, maximum, mean and standard deviation of the relative prediction error, over $5$ retrainings, for the best performing networks for each observable. As seen from this table, the sensitivity to retraining is rather low for the lift. On the other hand, it is more significant for the drag.
\item \emph{Variation of Network size.} We consider the dependence of performance with respect to variation of network size by varying the width and depth for the best performing networks. The resulting prediction errors are shown in table \ref{tab:afoilNS} where a $3 \times 3$ matrix for the errors (rows represent width and columns depth) is tabulated. All the network sizes are consistent with the requirements of lemma \ref{lem:1}. As observed from the table, the sensitivity to network size is greater for the drag than for the lift. It appears that increasing the width for constant depth results in lower errors for the drag. On the other hand, training clearly failed for some of the larger networks as the number of training samples appears to be not large enough to train these larger networks.
\end{itemize}
\subsubsection{Network performance with respect to number of training samples.}
 All the above results were obtained with $N=128$ training samples. We have repeated the entire ensemble training with different samples sizes i.e $N=32,64,256$ also and present a summary of the results in figure \ref{fig:afoilerrvsamp}, where we plot the mean (percentage relative) $L^2$ prediction error with respect to the number of samples. For each sample size, three different errors are plotted, namely the minimum (maximum) prediction errors and the mean prediction error over a subset of the hyperparameter space, that omits the \emph{bad set of ADAM} as the optimizer. Moreover, we only consider errors with selected (optimal) retrainings in each case. As seen from the figure, the prediction error \emph{decreases with sample size $N$} as predicted by the theory. The decay seems of the of form $\left(\frac{U_i}{N}\right)^{\alpha_i}$ with $\alpha_1=0.81$ and $U_1 = 0.77$ for the lift and $\alpha_2 = 0.92$ and $U_2 = 3.42$ for the drag. In both cases, the decay of error with respect to samples, is at a higher rate than that predicted by the theory in \eqref{eq:gerr2}. Moreover, the constants are much lower than the number of parameters in the network. Both facts indicate a very high degree of compressibility (at least in this possible pre-asymptotic regime) for the trained networks and explains the unexpectedly low prediction errors. 
\begin{figure}[htbp]
\subfigure[Lift]{\includegraphics[width=0.48\textwidth]{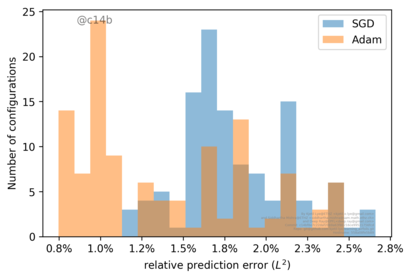}} 
\subfigure[Drag]{\includegraphics[width=0.48\textwidth]{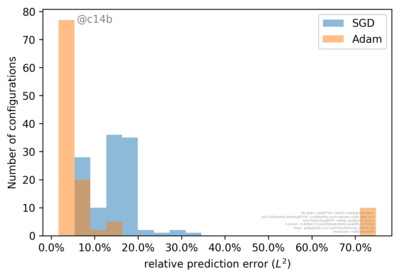}}
\caption{Histograms for the relative prediction errors comparing the ADAM and SGD algorithms on the flow past airfoils.}
\label{fig:afoilOpt}
\end{figure}
\begin{figure}[htbp]
\subfigure[Lift]{\includegraphics[width=0.48\textwidth]{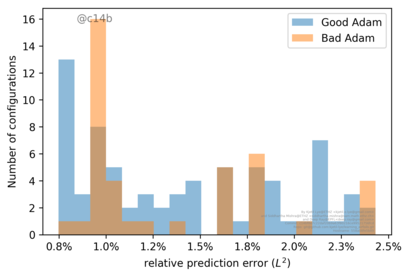}}
\subfigure[Drag]{\includegraphics[width=0.48\textwidth]{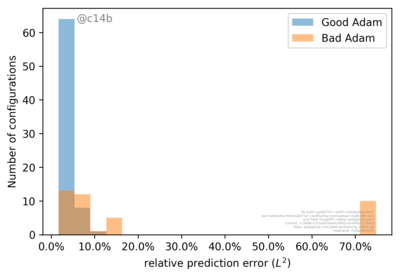}}
\caption{Histograms of the relative prediction errors, comparing \emph{Good} ADAM and \emph{bad} ADAM, for the flow past airfoils.}
\label{fig:afoilAdam}
\end{figure}

\begin{figure}[htbp]
\subfigure[Lift]{\includegraphics[width=0.48\textwidth]{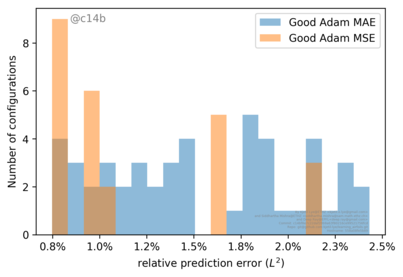}}
\subfigure[Drag]{\includegraphics[width=0.48\textwidth]{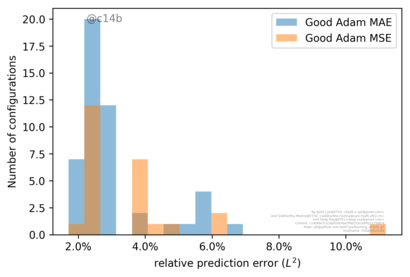}}
\caption{Histograms for the relative prediction errors, comparing the mean absolute error and mean square error loss functions, for the flow past airfoils}
\label{fig:afoillf}
\end{figure}
\begin{figure}[htbp]
\subfigure[Lift]{\includegraphics[width=0.48\textwidth]{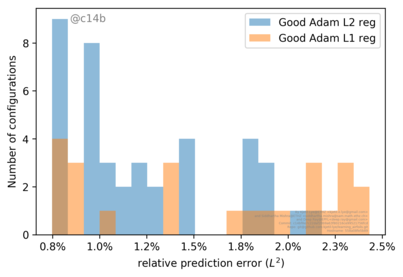}}
\subfigure[Drag]{\includegraphics[width=0.48\textwidth]{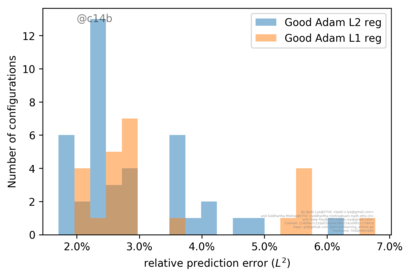}}
\caption{Histograms for the relative prediction error, comparing $L^1$ and $L^2$ regularizations of the loss function, for the flow past airfoils}
\label{fig:afoilreg}
\end{figure}
\begin{figure}[htbp]
\subfigure[Lift]{\includegraphics[width=0.32\textwidth]{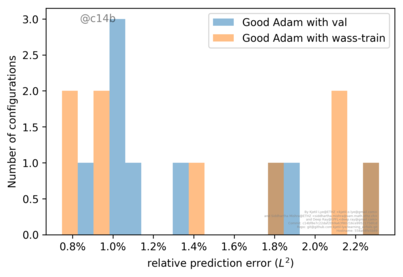}} 
\subfigure[Lift]{\includegraphics[width=0.32\textwidth]{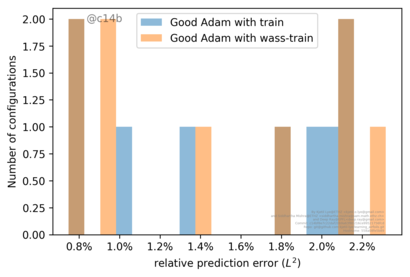}}
\subfigure[Lift]{\includegraphics[width=0.32\textwidth]{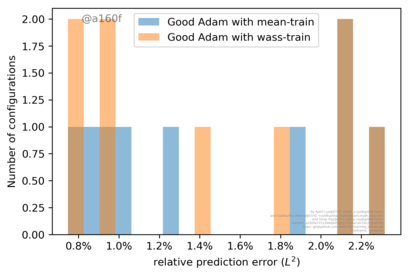}}
\subfigure[Drag]{\includegraphics[width=0.32\textwidth]{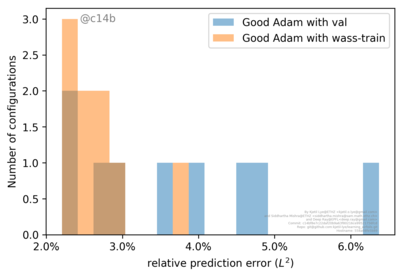}} 
\subfigure[Drag]{\includegraphics[width=0.32\textwidth]{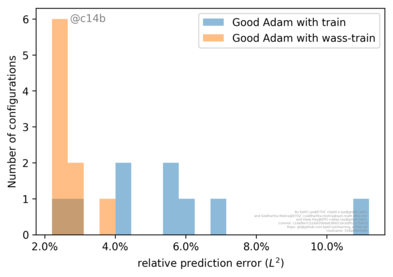}}
\subfigure[Drag]{\includegraphics[width=0.32\textwidth]{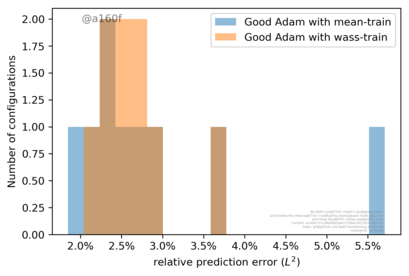}}
\caption{Network performance with respect to selection criteria for retrainings (see section \ref{sec:imp} i.e, \emph{train,val,mean-train,wass-train} for the flow past airfoils problem. We compare \emph{wass-train} with each of the other three criteria. Histograms for prediction error (X-axis) with number of configurations (Y-axis) are shown.}
\label{fig:afoilselect}
\end{figure}
\begin{table}[htbp]
\centering
\begin{tabular}{|l|l|l|l|l|}
\hline
Obs & Err (min) & Err (max) & Err (mean) & Err (std) \\
\hline
Lift  & $0.786$ & $1.336$ & $1.004$ & $0.187$ \\
\hline
Drag & $1.847$ & $8.016$ & $3.841$ & $2.332$ \\
\hline
\end{tabular}
\caption{Sensitivity of the best performing networks (listed in table \ref{tab:afoilBPNets}) to retrainings i.e starting values for ADAM for flow past airfoil. All errors are relative mean $L^2$ prediction error \eqref{eq:perr} in percentage and we list the minimum, maximum, mean and standard deviation of the error over $5$ retrainings.}
\label{tab:afoilretrain}
\end{table}
\begin{table}[htbp]
\centering
\begin{tabular}{|c|c|c|c|}
\hline
Width/Depth  & $6$  & $12$ & $24$ \\
\hline
$4$      & $\left[1.21,4.37 \right]$ &  $\left[0.83,1.74\right]$ &  $\left[1.09,2.48\right]$ \\
\hline 
$8$      & $\left[0.84,3.66 \right]$ &  $\left[0.86,2.02\right]$ &  $\left[0.81,1.69\right]$ \\
\hline
$16$      & $\left[0.80,\neg \right]$ &  $\left[\neg, 74.72\right]$ &  $\left[\neg,5.11\right]$ \\
\hline
\end{tabular}
\caption{Performance of trained neural networks with respect to the variation of network size for the flow past airfoils problem. The rows represent variation with respect to Width (number of neurons per hidden layer) and the 
columns represent variation with respect to Depth (number of hidden layers). For each entry of the Width-Depth matrix, we tabulate the vector of relative percentage mean prediction error for the best performing networks. The components of each vector represent error in $\left[{\rm Lift}, {\rm Drag} \right]$ and $\neg$ is used to indicate that the training procedure failed for this particular configuration.}
\label{tab:afoilNS}
\end{table}
\begin{figure}[htbp]
\subfigure[Lift]{\includegraphics[width=0.5\textwidth]{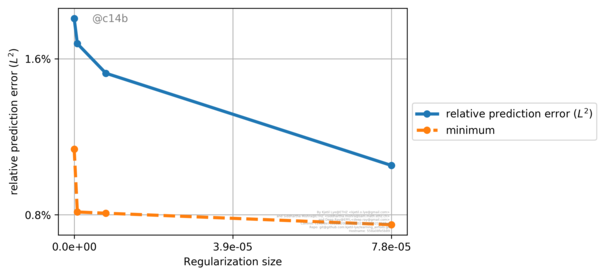}} 
\subfigure[Drag]{\includegraphics[width=0.5\textwidth]{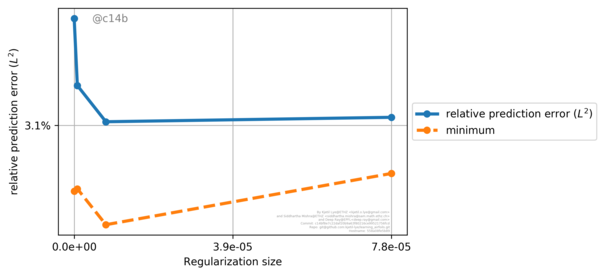}}
\caption{Variation of prediction error (Y-axis) with respect to the size of the regularization parameter $\lambda$ in \eqref{eq:lf2} (X-axis) for the flow past airfoils problems. The minimum and mean of prediction error (over all hyperparameter configurations) is shown}
\label{fig:afoilregval}
\end{figure}
\begin{figure}[htbp]
\subfigure[Lift]{\includegraphics[width=0.48\textwidth]{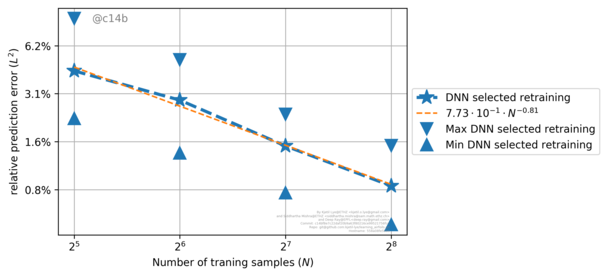}} 
\subfigure[Drag]{\includegraphics[width=0.48\textwidth]{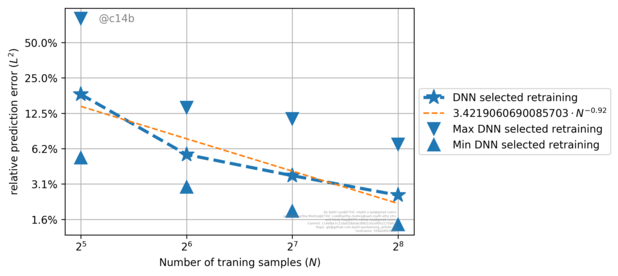}}
\caption{Relative percentage mean square prediction error \eqref{eq:perr} (Y-axis) for the neural networks approximating the flow past airfoils problem, with respect to number of training samples (X-axis). For each number of training samples, we plot the mean, minimum and maximum (over hyperparameter configurations) of the prediction error. Only the selected (optimal) retraining is shown. }
\label{fig:afoilerrvsamp}
\end{figure}
\subsubsection{UQ with deep learning}
Next, we approximate the underlying probability distributions (measures) \eqref{eq:pf1}, with respect to the lift and the drag in the flow past airfoils problem with the DLQMC algorithm \ref{alg:dlqmc}. A reference measure is computed from the whole test set of $1001$ samples and the corresponding histograms are shown in figure \ref{fig:afoilhistcomp} to visualize the probability distributions. In the same figure, we also plot corresponding histograms, computed with the best performing networks for each observable (listed in table \ref{tab:afoilBPNets}) and observe very good qualitative agreement between the outputs of the DLQMC algorithm and the reference measure.

In order to quantify the gain in computational efficiency with the DLQMC algorithm over the baseline QMC algorithm in approximating probability distributions, we will compute the \emph{speedup}. To this end, we observe from figure \ref{fig:afoilQMCcomp} that the baseline QMC algorithm converges to the reference probability measure with respect to the Wasserstein distance \eqref{eq:m2}, at a rate of $\alpha=0.81$, for both the lift and drag. Next, from table \ref{tab:afoilcost}, we see that once the training samples have been generated, the cost of training the network and performing evaluations with the trained network is essentially free. Consequently, if the number of training samples is $N$, we compute a \emph{raw speed-up} (for each observable) given by,
\begin{equation}
\label{eq:rsp1}
\sigma_{i,dlqmc}^{raw} := \frac{W_1\left(\hat{\mu}_{i,qmc,N},\hat{\mu_i}^{\Delta} \right)}{W_1\left(\hat{\mu}^{\ast}_{i,qmc},\hat{\mu_i}^{\Delta} \right)}.
\end{equation}
Here for $i={\rm Lift},~{\rm Drag}$, $\hat{\mu}_i^{\Delta}$ is the reference measure (probability distribution) computed from the test set $\test$, $\hat{\mu}_{i,qmc,N}$ is the measure \eqref{eq:m1} with the baseline QMC algorithm and $N$ QMC points, and $\hat{\mu}^{\ast}_{i,qmc}$ is the measure computed with the DLQMC algorithm \eqref{eq:dlqmc1}. The real speed up is then given by,
\begin{equation}
\label{eq:esp1}
\sigma_{i,dlqmc}^{real} = \left(\sigma_{i,dlqmc}^{raw}\right)^{\frac{1}{\alpha_i}},
\end{equation}
in order to compensate for the rate of convergence of the baseline QMC algorithm. In other words, our \emph{real speed up} compares the costs of computing errors in the Wasserstein metric, that are of the same magnitude, with the DLQMC algorithm vis a vis the baseline QMC algorithm. Henceforth, we denote $\sigma_{dlqmc}^{real}$ as the \emph{speedup} and list the speedups (for each observable) with the DLQMC algorithm for $N=128$ training samples in table \ref{tab:afoilrsp}. We observe from this table that  we obtain speedups ranging between half an order to an order of magnitude for the lift and the drag. The speedups for drag are slightly better than that for lift. 

In figure \ref{fig:afoilrspvsN}, we plot the speedup of the DLQMC algorithm (over the baseline QMC algorithm) as a function of the number of training samples in the deep learning algorithm \ref{alg:DL}. As seen from the figure, the best speed ups are obtained in case of $64$ training samples for the lift and $128$ training samples for the drag. This is on account of a complex interaction between the prediction error which decreases (rather fast) with the number of samples and the fact that the errors with the baseline QMC algorithm also decay with an increase in the number of samples. 
\subsubsection{Comparison with Monte Carlo algorithms}
We have computed a Monte Carlo approximation of the probability distributions of the lift and the drag by randomly selecting $N=128$ points from the parameter domain $Y = [0,1]^6$ and computing the probability measure $\hat{\mu}_{mc} = \frac{1}{N} \sum_{j=1}^N \delta_{\map^{\Delta}(y_j)}$. We compute the Wasserstein error $W_1(\hat{\mu}_{mc}, \hat{\mu}^{\Delta})$, with respect to the reference measure $\hat{\mu}^{\Delta}$ and divide it with the error obtained with the DLQMC algorithm to obtain a \emph{raw speedup} of the DLQMC algorithm over MC. As MC errors converge as a square root of the number of samples, the \emph{real speedup} of DLQMC over MC is calculated by squaring the raw speedup. We present this real speedup over MC in table \ref{tab:afoilmcsp} and observe from that the speedups of the DLQMC algorithm over the baseline MC algorithm are very high and amount to at least two orders of magnitude. This is not surprising as the baseline QMC algorithm significantly outperforms the MC algorithm for this problem. Given that the DLQMC algorithm was an order of magnitude faster than the baseline QMC algorithm, the cumulative effect leads to a two orders of magnitude gain over the baseline MC algorithm. 
\subsubsection{Choice of training set.}
 For the sake of comparision with our choice of Sobol points to constitute the training set $\train$ in this problem, we chose $N=128$ random points in $Y$ as the training set.  With this training set, we repeated the ensemble training procedure and found the best performing networks. The relative percentage mean $L^2$-prediction error \eqref{eq:perr} (with respect to a test set of $320$ random points) was computed and is presented in table \ref{tab:MCtrain}. As seen from this table, the prediction errors with respect to the best performing networks are considerably (an order of magnitude) higher than the corresponding errors obtained for the QMC training points (compare with table \ref{tab:afoilBPNets}). Hence, at least for this problem, Sobol points are a much better choice for the training set than random points. An intuitive reason for this could lie in the fact that the Sobol points are better distributed in the parameter domain than random points. So on an average, a network trained on them will generalize better to unseen data, than it would if trained on random points. 
 \begin{figure}[htbp]
\subfigure[Lift]{\includegraphics[width=0.45\textwidth]{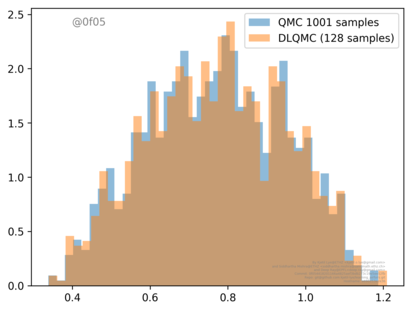}} 
\subfigure[Drag]{\includegraphics[width=0.45\textwidth]{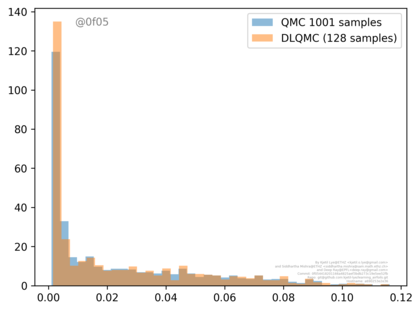}}
\caption{Empirical histograms representing the probability distribution (measure) for the lift and the drag in the flow past airfoils problem. We compare the reference histograms (computed with the test set $\test$) and the histograms computed with the DLQMC algorithm.}
\label{fig:afoilhistcomp}
\end{figure}
\begin{figure}[htbp]
	\subfigure[Lift]{\includegraphics[width=0.45\textwidth]{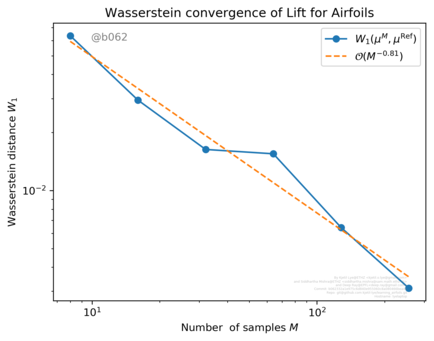}}
\subfigure[Drag]{\includegraphics[width=0.45\textwidth]{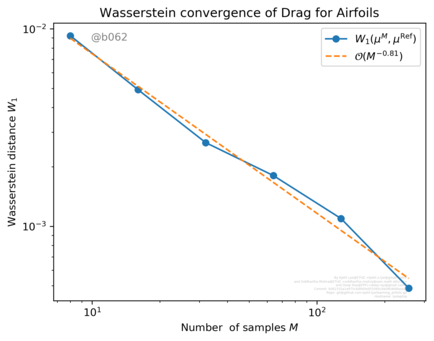}}
\caption{Convergence of the Wasserstein distance $W_1\left(\hat{\mu}_{i,qmc},\hat{\mu_i}^{\Delta} \right)$ (Y-axis) for the of the lift and the drag for the flow past airfoils, with respect to number of QMC points (X-axis)}
\label{fig:afoilQMCcomp}
\end{figure}
\begin{table}[htbp]
\centering
\begin{tabular}{|l|l|}
\hline
Observable & Speedup \\
\hline
Lift & $6.64$ \\
\hline
Drag & $8.56$ \\
\hline

\end{tabular}
\caption{Real speedups \eqref{eq:esp1}, for the lift and the drag in the flow past airfoils problem, comparing DLQMC with baseline QMC algorithm.}
\label{tab:afoilrsp}
\end{table}
\begin{table}[htbp]
\centering
\begin{tabular}{|l|c|}
\hline
Observable & Speedup over MC \\
\hline
Lift & $246.02$ \\
\hline
Drag & $179.54$ \\
\hline
\end{tabular}
\caption{Real speedups, for the lift and the drag in the flow past airfoils problem, comparing DLQMC with the baseline MC algorithm.}
\label{tab:afoilmcsp}
\end{table}
 
\begin{figure}[htbp]
\subfigure[Lift]{\includegraphics[width=0.45\textwidth]{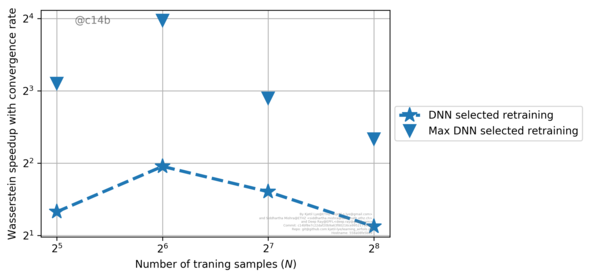}} 
\subfigure[Drag]{\includegraphics[width=0.45\textwidth]{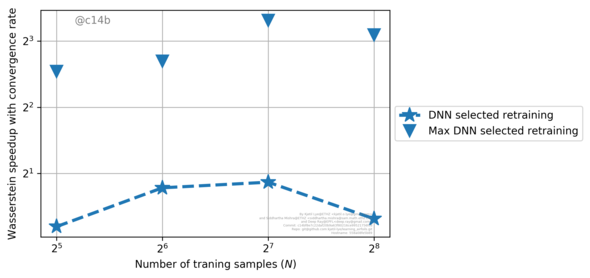}}
\caption{Real speedups \eqref{eq:esp1} for the DLQMC algorithm over the baseline QMC algorithm (Y-axis) with respect to number of training samples (X-axis) for the flows past airfoils problem. The maximum and mean speed up (over all hyperparameter configurations)  are shown.}
\label{fig:afoilrspvsN}
\end{figure}
\begin{table}[htbp]
\centering
\begin{tabular}{|l|l|l|l|l|l|l|}
\hline
Obs & Opt & Loss & $L^1$-reg & $L^2$-reg & Selection & \emph{BP}. Err mean  \\
\hline
Lift & SGD & MAE & $0.0$ & $7.8\times 10^{-7}$ & mean-train & $8.487$    \\
\hline
Drag & ADAM & MAE & $7.8 \times 10^{-7}$ &$0.0$ &train & $20.25$  \\
\hline 
\end{tabular}
\caption{The hyperparameter configurations that correspond the best performing network (one with least mean prediction error) for the two observables of the flow past airfoils, trained on random (Monte Carlo) training set.}
\label{tab:MCtrain}
\end{table}
\subsection{A stochastic shock tube problem}
As a second numerical example, we consider a shock tube problem with the one-dimensional version of the compressible Euler equations \eqref{eq:euler2D}.  

\subsubsection{Problem description}
The underlying computational domain is $[-5,5]$ and initial conditions are prescibed in terms of left state $ (\rho_L, \ u_L, \ p_L)$ and a right state $ (\rho_R, \ u_R, \ p_R)$ at the initial discontinuity $x_0$. These states are defined as,
\begin{equation}
    \label{eq:stin}
    \begin{aligned}
\rho_L&= 0.75+0.45G_1(y) \quad &\rho_R = 0.4 + 0.3G_2(y) \\
w_L&= 0.5+0.5G_3(y) \quad &w_R=0\\ 
p_L &=2.5+1.6G_4(y) \quad &p_R= 0.375+0.325G_5(y) \\
x_0 &= 0.5G_6(y).
\end{aligned}
\end{equation}
$y \in Y = [0,1]^6$ is a parameter and $G_k(y) = 2y_k-1$ for $k=1,...,6$. We use the Lebesgue measure on $Y$ as the underlying measure. 

The observables for this problem are chosen as the average integral of the density over fixed intervals
\begin{equation}\label{eq:sod_obs}
\map_j(y) = \frac{1}{|I_j|}\int_{I_j} \rho(x,T_f;y) dx, \quad j =1,2,3,
\end{equation}
where $I_1 = [-1.5,-0.5], I_2 = [0.8,1.8], I_3 = [2,3]$ are the underlying regions of interest. 

A priori, this problem appears to be simpler as it is only one-dimensional when compared to the two-dimensional airfoil. However, the difficulty stems from the \emph{large variance} of the initial data. As seen from \eqref{eq:stin}, the ratio of the standard deviation to the mean, for the underlying variables, is very high and can be as high as $100 \%$. This is almost an order of magnitude larger than the corresponding initial stochasticity of the flow past airfoil problem and we expect this large initial variance to be propagated into the solution. This is indeed borne out by the results shown in figure \ref{fig:STinit}, where we plot the mean (and standard deviation) of the density, both initially and at time $T_f = 1.5$. Moreover, the initial stochastic inputs span a wide range of possible scenarios. Particular choices of the parameters i.e, $(y_1,y_2,y_3,y_5,y_6) = (7/9,1/24,0,1/32,1/13,0)$ and $(y_1,y_2,y_3,y_5,y_6) = (0.1611,2/3,0.698,0.8213,0.8015,0)$, correspond to the classical Sod shock tube and Lax shock tube problems \cite{FMT}. Thus, we can think of this stochastic shock tube problem as more or less generic (scaled) Riemann problem for the Euler equations.

In a recent paper \cite{LMM1}, the authors have analyzed the role of the underlying variance in the generalization errors for deep neural networks trained to regress functions. In particular, high variance can lead to large generalization errors. In this sense, this shock tube problem might be harder to predict with neural networks than the flow past airfoil. 

\begin{figure}[htbp]
\subfigure[$T=0$]{\includegraphics[width=0.45\textwidth]{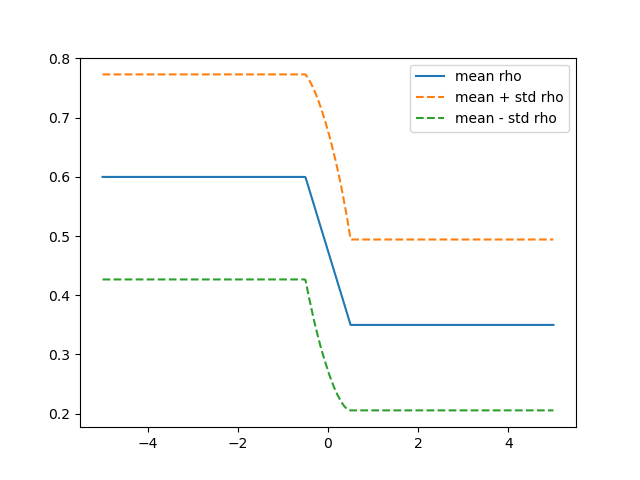}} 
\subfigure[$T=1.5$]{\includegraphics[width=0.45\textwidth]{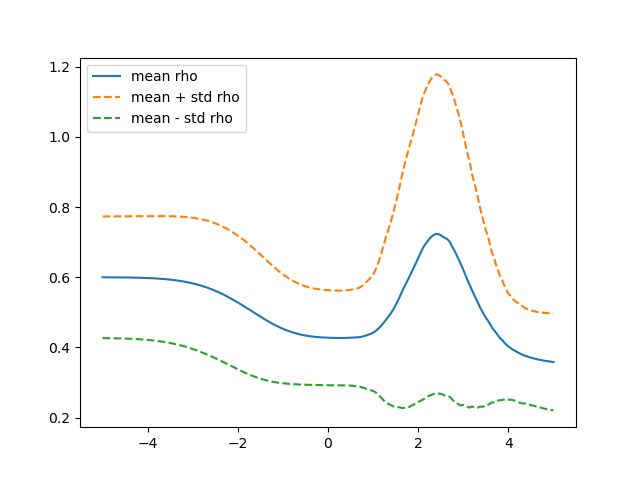}} 
\caption{Mean ($\pm$ standard deviation) for the density of the stochastic shock tube problem.}
\label{fig:STinit}
\end{figure}

\subsubsection{Generation of training data.}
In order to generate the training and test sets, we denote $\EuScript{Q}^{sob}$ as the first $16256$ Sobol points on the domain $[0,1]^6$. The training set $\train$ consists of the first $128$ of these Sobol points whereas as the test set is $\test = \EuScript{Q}^{sob}$. For each point $y_j \in\EuScript{Q}^{sob} $, the maps or (rather their numerical surrogate) $\map_{1,2,3}^{\Delta}(y_j)$ are generated from high-resolution numerical approximations $\U^\Delta(y)$ obtained using a second-order finite volume scheme. The domain is discretized into a uniform mesh of $K=2048$ disjoint intervals, with the HLL numerical flux and with the left and right cell-interface values obtained using a second-order WENO reconstruction. Time marching is performed using the second-order Runge-Kutta method with a CFL number of 0.475.  Once the finite volume solution is computed, the corresponding observables \eqref{eq:sod_obs} at each $y$ are approximated as
\begin{equation}\label{eq:sod_obs_approx}
\map^\Delta_j(y) = \frac{\sum_{i \in \Lambda_j} \rho_i(y)}{ \#(\Lambda_j)} , \quad  \Lambda_j = \{i \ | \ x_i \in I_j\},\quad j =1,2,3,
\end{equation}
where $x_i$ is the barycenter of the cell $\om_i$.
\subsubsection{Results.}
We follow the ensemble training procedure outlined in section \ref{sec:imp}. The network architecture, specified in table \ref{tab:afoilRefNet} is used together with the $114$ hyperparameter configurations of the previous section, with $5$ retrainings for each hyperparameter configuration. The best performing networks for each observable are identified as the ones with lowest prediction error \eqref{eq:perr} and these networks are presented in table \ref{tab:stBPNets}. From the table, we observe very slight differences for the best performing networks for the three observables. However, these networks are indeed different from the best performing networks for the flow past airfoil (compare with table \ref{tab:afoilBPNets}). Moreover, the predictions errors are larger than the flow past airfoil case, ranging between $2\%$ for the observable $\map_1$ to approximately $6\%$ for the other two observables. This larger error is expected as the variance of the underlying maps is significantly higher. Nevertheless, a satisfactory prediction accuracy is obtained for only $128$ training samples. 
\begin{table}[htbp]
\centering
\begin{tabular}{|l|l|l|l|l|l|l|}
\hline
Obs & Opt & Loss & $L^1$-reg & $L^2$-reg & Selection & \emph{BP}. Err mean  \\
\hline
$\map_1$ & ADAM & MSE & $7.8\times 10^{-6}$ & $0$ & wass-train & $2.212$ \\
\hline
$\map_2$ & ADAM & MSE & $7.8 \times 10^{-6}$ & $0$ & wass-train & $6.575$  \\
\hline 
$\map_3$ & ADAM & MSE & $7.8 \times 10^{-5}$ & $0$ & mean-train & $5.467$  \\
\hline
\end{tabular}
\caption{The hyperparameter configurations that correspond to the best performing network (one with least mean prediction error) for the three observables of the stochastic shock tube problem.}
\label{tab:stBPNets}
\end{table}
Results of the ensemble training procedure are depicted in the form of histograms for the prediction errors in figure \ref{fig:stHistErr}. From this figure, we observe that although there are some outliers, most of the hyperparameter configurations resulted in errors comparable to the best performing networks (see table \ref{tab:stBPNets}). Furthermore, a large majority of the hyperparameters resulted in substantially (a factor of $3-4$) smaller prediction errors, when compared to the linear least squares regression \eqref{eq:llsq}. 

A sensitivity study (not presented here) indicates very similar robustness results as for the flow past the airfoil. In figure \ref{fig:sterrvsamp}, we plot the prediction error with respect to the number of samples and observe that the mean error (over hyperparameter configurations) decays with an exponent of $0.5$ for the observables $\map_{2,3}$. On the other hand, the decay is a bit slower for $\map_1$. However, the low value of the constants still indicates significant compression allowing us to approximate this rather intricate problem with accurate neural networks. 

Finally, we compute the underlying probability distributions for each observable with the DLQMC algorithm \ref{alg:dlqmc}. In figure \ref{fig:stUQ}, we plot the maximum and mean real speedup \eqref{eq:esp1}, over hyperparameter configurations, for each observable, with respect to the number of training samples. From this figure, we observe maximum speedups of about $16$ for the observables $\map_{1,2}$ and $5$ for the observable $\map_3$, indicating that the DLQMC algorithm is significantly superior to the baseline QMC algorithm, even for this problem. 
\begin{figure}[htbp]
\subfigure[$\map_1$]{\includegraphics[width=0.3\textwidth]{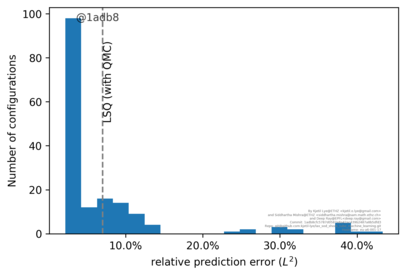}} 
\subfigure[$\map_2$]{\includegraphics[width=0.3\textwidth]{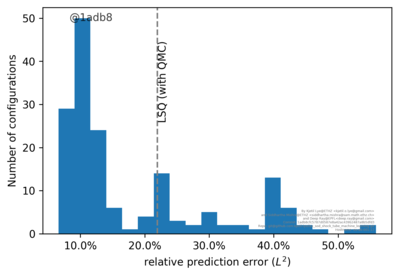}}
\subfigure[$\map_3$]{\includegraphics[width=0.3\textwidth]{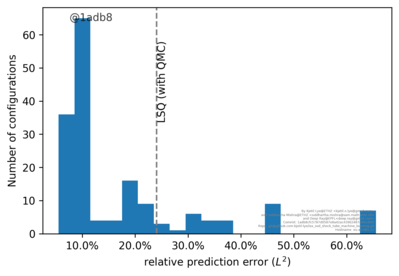}}
\caption{Histograms depicting distribution of the percentage relative mean $L^2$ prediction error \eqref{eq:perr} (X-axis) over number of hyperparameter configurations (samples, Y-axis)  for the observables of the stochastic shock-tube problem.}
\label{fig:stHistErr}
\end{figure}

\begin{figure}[htbp]
\subfigure[$\map_1$]{\includegraphics[width=0.3\textwidth]{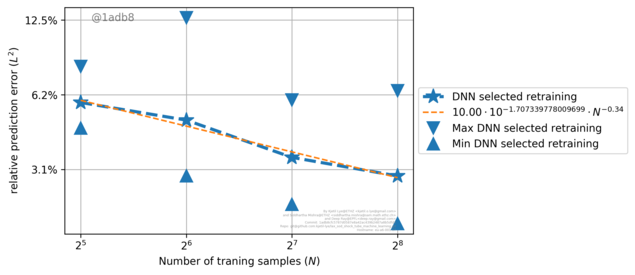}} 
\subfigure[$\map_2$]{\includegraphics[width=0.3\textwidth]{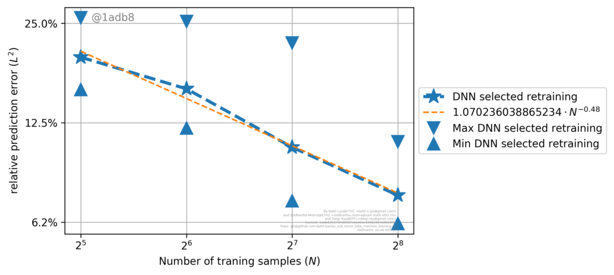}}
\subfigure[$\map_3$]{\includegraphics[width=0.3\textwidth]{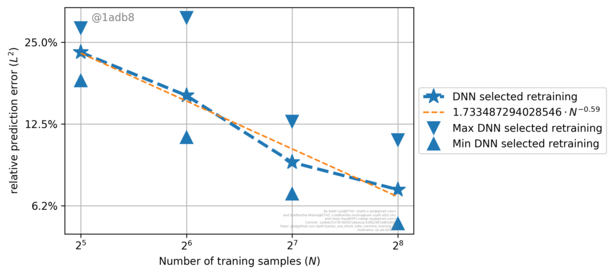}}
\caption{Relative percentage mean square prediction error \eqref{eq:perr} (Y-axis) for the neural networks approximating the stochastic shock tube problem, with respect to number of training samples (X-axis). For each number of training samples, we plot the mean, minimum and maximum (over hyperparameter configurations) of the prediction error. Only the selected (optimal) retraining is shown. }
\label{fig:sterrvsamp}
\end{figure}

\begin{figure}[htbp]
\subfigure[$\map_1$]{\includegraphics[width=0.3\textwidth]{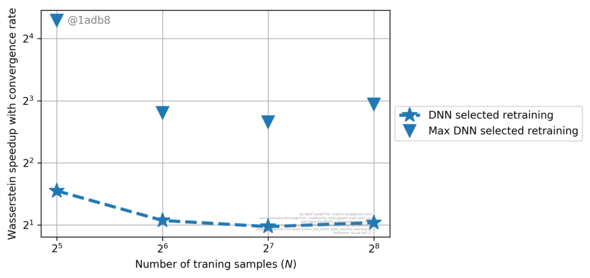}} 
\subfigure[$\map_2$]{\includegraphics[width=0.3\textwidth]{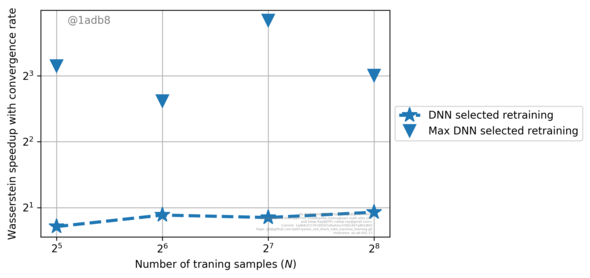}} 
\subfigure[$\map_3$]{\includegraphics[width=0.3\textwidth]{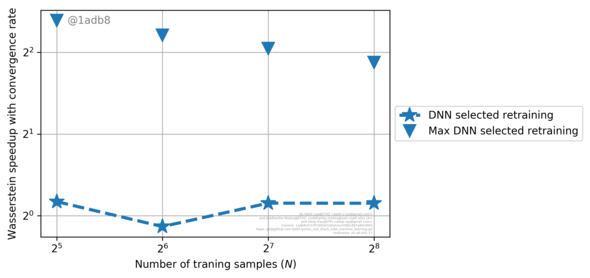}} 
\caption{Real speedups \eqref{eq:esp1} for the DLQMC algorithm over the baseline QMC algorithm (Y-axis) with respect to number of training samples (X-axis) for the stochastic shock tube problem. The maximum and mean speed up (over all hyperparameter configurations)  are shown.}
\label{fig:stUQ}
\end{figure}
\section{Discussion}
\label{sec:disc}
Problems in CFD such as UQ, Bayesian Inversion, optimal design etc are of the \emph{many query} type i.e, solving them requires evaluating a large number of realizations of the underlying input parameters to observable map \eqref{eq:ptoob}. Each realization involves a call to an expensive CFD solver. The cumulative total cost of the many-query simulation can be prohibitively expensive. 

In this paper, we have tried to harness the power of machine learning by proposing using \emph{deep fully connected neural networks}, of the form \eqref{eq:ann1}, to learn and predict the underlying parameters to observable map. However, a priori, this constitutes a formidable challenge. In section \ref{sec:theo}, we argue by a combination of theoretical considerations and numerical experiments that the crux of this challenge it to find neural networks to learn maps of \emph{low regularity in a data poor regime}. We are in this regime as the many parameters to observable maps \eqref{eq:ptoob} in fluid dynamics, can be at best, Lipschitz continuous. Moreover, given the cost, we can only afford to compute a few training samples. 

We overcome these difficulties with the following \emph{novel} ideas, first,  we chose to focus on learning observables rather than the full solution of the parametrized PDE \eqref{eq:cdep}. Observables are more regular than fields (see section \ref{sec:theo}) and might have less underlying variance and are thus, easier to learn, see \cite{LMM1} for a more recent theoretical justification. Second, we propose using \emph{low-discrepany sequences} to generate the training set $\train$. The equi-distribution property of these points might ensure lower generalization errors than those resulting from training data generated at random points, see the forthcoming paper \cite{RM1} for further details. Finally, we propose an \emph{ensemble training procedure} (see section \ref{sec:imp}) to systematically scan the hyperparameter space in order to winnow down the best performing networks as well as to test sensitivity of the trained networks to hyperparameter choices. 

We presented two prototypical numerical experiments for the compressible Euler equations in order to demonstrate the efficacy of our approach. We observed from the numerical experiments that the trained networks were indeed able to provide low prediction errors, despite being trained on very few samples. Moreover, a large number of hyperparameter configurations were close in performance to the best performing networks, indicating a significant amount of robustness to most hyperparameters. More crucially, the obtained generalization errors followed the theoretical predictions of section \ref{sec:theo} and indicating a high degree of compression, explaining why the networks generalized well. It is essential to state that these low prediction errors are obtained even when the cost of evaluating the neural network is several orders of magnitude lower than the full CFD solve. 

As a concrete application of our deep learning algorithm \ref{alg:DL}, we combined it with a Quasi-Monte Carlo (QMC) method to propose a deep learning QMC algorithm \ref{alg:dlqmc}, for performing forward UQ, by approximating the underlying probability distribution \eqref{eq:pf1} of the observable. In theorem \ref{thrm:4}, we prove that the DLQMC algorithm is guaranteed to out-perform the baseline QMC algorithm as long as the underlying neural network generalizes well. This is indeed borne out in the numerical experiments, where we observed about an order of magnitude speedup of the DLQMC algorithm over the corresponding QMC algorithm and two orders of magnitude speedup over a Monte Carlo algorithm. 

Thus, we have demonstrated the viability of using neural networks to learn observables and to solve challenging problems of the many-query type. Our approach can be compared to a standard model order reduction algorithm \cite{MORbook}. Here, generating the training data and training the networks is the \emph{offline} step whereas evaluating the network constitutes the \emph{online} step. Our results compare favorably with attempts to use Model order reduction for hyperbolic problems \cite{AAC,Cris1}, particularly when it comes to the low costs of training and evaluation. On the other hand, typical MOR algorithms provide a surrogate for the solution field, where we only predict the observable of interest in the approach presented here. 

Our results can be extended in many directions; we can adapt this setting to learn the full solution field, instead of the observables. However, this might be harder in practice as the field is even less regular than the observable. Our approach is entirely data driven, in the sense that we do not use any specific information about the underlying parameterized PDE \eqref{eq:cdep}. Thus, the approach can be readily extended to other forms of \eqref{eq:cdep}, for instance the Navier-Stokes equations or other elliptic and parabolic PDEs. Finally, we plan to apply the trained neural networks in the context of Bayesian inverse problems and shape optimization under uncertainty, in forthcoming papers.

\section*{Acknowledgements}
The research of SM is partially support by ERC Consolidator grant (CoG) NN 770880 COMANFLO. A large proportion of computations for this paper were performed on the
ETH compute cluster EULER.

\end{document}